\documentclass[draftcls,onecolumn,12pt]{IEEEtran}
\usepackage{color}
\usepackage{verbatim}
\usepackage{amsfonts}
\usepackage{amssymb}
\usepackage{stfloats}
\usepackage{cite}
\usepackage{graphicx}
\usepackage{psfrag}
\usepackage{subfigure}
\usepackage{amsmath}
\usepackage{array}
\usepackage{epstopdf}
\usepackage{authblk}
\usepackage{graphicx} 
\usepackage{amsthm} 
\usepackage{lipsum}
\usepackage{verbatim} 
\usepackage{authblk}
\usepackage{mathtools}
\usepackage{cuted}
\usepackage[lined,boxed,ruled]{algorithm2e}
\usepackage{algpseudocode}
\usepackage{framed} 
\usepackage{subfigure}
\usepackage{soul}
\usepackage{bm}
\usepackage{setspace}
\usepackage{url}

\newtheorem{theorem}{Theorem}

\newtheorem{lemma}{Lemma}

\newtheorem{proposition}{Proposition}

\newtheorem{assumption}{Assumption}

\newtheorem{remark}{\bf Remark}
\def\phi{\varphi}

\def\({\left(}
\def\){\right)}

\def\bee{{\mathbf{e}}}
\def\bff{{\mathbf{f}}}
\def\bg{{\mathbf{g}}}

\def\b0{{\mathbf{0}}}

\usepackage{tabularx}

\newcounter{protocol}

\newtheorem{Example}{Example}

\title{Over-the-Air Multi-View Pooling for Distributed Sensing}

\author{Zhiyan~Liu, Qiao~Lan, Anders~E.~Kalør, Petar~Popovski, and Kaibin~Huang
\thanks{Z. Liu, Q. Lan and K. Huang are with Department of Electrical and Electronic Engineering at The University of Hong Kong (HKU), Hong Kong. A. E. Kalør and P. Popovski are with Department of Electronic Systems, Aalborg University, Aalborg, Denmark. A. E. Kalør is also affiliated with Dept. of EEE at HKU. The work of A. E. Kalør was supported by the Independent Research Fund Denmark (IRFD) under Grant 1056-00006B. The work of P. Popovski was supported by the Villum Investigator Grant “WATER” from the Velux Foundation, Denmark. Contact: K. Huang  (Email: huangkb@eee.hku.hk).}}

\makeatletter
\newcommand{\removelatexerror}{\let\@latex@error\@gobble}
\makeatother

\IEEEoverridecommandlockouts

\begin{document}

\maketitle

\vspace{-15mm}
\begin{abstract}
    Sensing is envisioned as a key network function of the \emph{sixth-generation} (6G) mobile networks. \emph{Artificial intelligence} (AI)-empowered sensing {fuses} features of multiple sensing views from devices distributed in edge networks for the edge server to perform accurate inference. {This process, known as \emph{multi-view pooling}}, creates a communication bottleneck due to multi-access by many devices. To alleviate this issue, we propose a task-oriented simultaneous access scheme for distributed sensing called \emph{Over-the-Air Pooling} (AirPooling). The existing \emph{Over-the-Air Computing} (AirComp) technique can be directly applied to enable Average-AirPooling, which exploits the waveform superposition property of a multi-access channel to implement fast over-the-air averaging of pooled features. 
    {However, despite being most popular in practice, the over-the-air maximization, called Max-AirPooling, is not AirComp realizable given the fact that AirComp addresses only a limited subset of functions.} 
    We tackle the challenge by proposing the novel generalized AirPooling framework that can be configured to support both Max- and Average-AirPooling by controlling a configuration parameter {and extended to even other pooling functions}. The former is realized by adding to AirComp the designed pre-processing at devices and post-processing at the server. To characterize the \emph{End-to-End} (E2E) sensing performance in object recognition, the theory of classification margin is applied to relate the classification accuracy and the AirPooling error, which allows the latter to be a tractable surrogate of the former. Furthermore, the analysis reveals an inherent tradeoff of Max-AirPooling between the accuracy of the pooling-function approximation and the effectiveness of noise suppression. Using the tradeoff, we make an attempt to optimize the configuration parameter of Max-AirPooling, yielding a sub-optimal closed-form method of adaptive parametric control. Experimental results obtained on real-world datasets show that AirPooling provides sensing accuracies close to those achievable by the traditional digital air interface but dramatically reduces the communication latency, by up to an order of magnitude.
\end{abstract}

\begin{IEEEkeywords}
\vspace{-3mm}
Edge inference, distributed sensing, multiple access, over-the-air computation.
\end{IEEEkeywords}

\section{Introduction}
The \emph{sixth-generation} (6G) mobile networks will feature at least two new functions~\cite{Huawei2022}. 
One is the ubiquitous deployment of \emph{artificial intelligence} (AI) algorithms at the network edge, termed edge AI, to empower \emph{Internet-of-things} (IoT) applications~\cite{Bennis2019ProcIEEE,GX2020CM}. The other new function is large-scale distributed sensing via cross-network collaboration between edge devices~\cite{Liu2022JSAC,Huawei2022}.
The natural integration of edge AI and network sensing, known as \emph{AI-of-Things} (AIoT) sensing, combines the strengths of multi-view observations by sensors and the powerful prediction capabilities of deep neural network models to make sensing accurate and intelligent~{\cite{Shen2022CST}}. This provides a platform for automating  wide-ranging applications including e-healthcare, autonomous driving, smart cities, environment monitoring, and automated manufacturing.
{ Consider an intelligent transportation system with split inference for example. A roadside server aggregates and pools view features extracted by on-vehicle sub-models on nearby vehicles  and infers the current traffic situation using the server sub-model.}
However, the deployment of AIoT sensing is stymied by a communication bottleneck caused by the need of uploading high-dimensional features extracted from sensing data at many sensors to a server for aggregation and inference. 
This bottleneck motivates the current work that presents a task-oriented multi-access framework, called \emph{over-the-air pooling}
(AirPooling) of simultaneously transmitted features, which provides a scalable air interface for distributed AIoT sensing.

AIoT sensing builds on an architecture known as \emph{split inference}~\cite{Huawei2022}. Essentially, a trained model is split into a low-complexity sub-model at a device and a deep sub-model at a server. The former extracts feature maps from raw data while the latter performs inference on the uploaded features. 
Such an architecture provides resource-constrained edge devices access to large-scale AI models at servers (e.g., image recognition with tens-to-hundreds of object classes) while protecting their data ownership~\cite{Zhang2020CM}.
One research focus in split inference is  on task-oriented communications aiming to optimize the \emph{end-to-end} (E2E) inference throughput, accuracy, or latency. 
To this end, researchers have designed a range of relevant techniques~\cite{Chen2020TWC,Deniz2021JSAC,Zhang2020ICC,Zhang2020CM}. To overcome the communication constraints, the splitting point of a model can be adapted to the available bandwidth and latency requirement \cite{Chen2020TWC}. On the other hand, for optimization of E2E system performance, researchers have developed the popular approach of \emph{joint source-channel coding} for split inference. It consists of  a pair of jointly designed neural-network encoder and decoder, where the former is used at a transmitter to map the source (e.g., features of images) to channel symbols and the latter at a receiver to perform on noisy channel outputs joint channel decoding and computation (e.g., image reconstruction or inference) \cite{Deniz2021JSAC,Zhang2020ICC}. In addition, feature quantization is another aspect of communication-efficient split inference. Researchers have applied the algorithm of variational information bottleneck to designing a channel-adaptive feature quantizer that aims to minimize the communication overhead for classification given a time-varying channel~\cite{Zhang2020CM}.

The split inference for distributed AIoT sensing can be realized by adding to the popular \emph{Multi-View Convolutional Neural Network} (MVCNN) architecture an air interface between multiple sensors and a server (fusion center) \cite{Hang2015ICCV}. Compared with point-to-point split inference~\cite{Chen2020TWC,Deniz2021JSAC,Zhang2020ICC,Zhang2020CM}, the distinctive feature of MVCNN is \emph{multi-view pooling}, referring to the fusion of features extracted from different sensors' views into a global feature map that is fed into the server's inference model. Two common types of multi-view pooling operation are Average- and Max-Pooling that compute the average and maximum of distributed features, respectively \cite{Hang2015ICCV,LYC2020CVPR,Chen2021TPAMI}. When there are many sensors,  the existing techniques for distributed compression and scheduling can be useful in overcoming the communication bottleneck. Most recently, a method of distributed information bottleneck has been proposed to optimize the tradeoff between the communication rate and the distortion of the prediction results \cite{SJW2022TWC}. Regarding scheduling, a sensor selection protocol is proposed in~\cite{LYC2020CVPR} where a sensor is selected if its observation is sufficiently correlated with the query from the fusion center. In view of prior work, the  task-oriented multi-access designs for AIoT sensing are still a largely unexplored area. 

Contributing to this area, the proposed AirPooling is a simultaneous-access  technique for realizing multi-view pooling over-the-air by exploiting the waveform-superposition property of a multi-access channel. AirPooling belongs to a class of techniques called \emph{over-the-air computation} (AirComp) for efficient wireless data aggregation~\cite{GX2021WCM}. The main motivation driving AirComp research is to overcome channel distortion and noise such that data averaging or other computation functions can be implemented accurately over-the-air. This gives rise to a rich set of relevant techniques such as power control (see, e.g., \cite{ Xiaowen2020TWC,CaoXW2022JSAC}), sub-channel selection (see, e.g.,~\cite{GX2020TWC}), and multi-antenna beamforming (see, e.g.,~\cite{Shi2020TWC,GX2019IOTJ}). Most recently, AirComp sees growing popularity in its application to supporting efficient model/gradient aggregation in federated learning, known as over-the-air federated learning \cite{Deniz2020TSP,Eldar2021TSP,Shi2020TWC,GX2021TWC}. Researchers also proposed the use of AirComp to realize majority-voting over-the-air in a distributed inference system~\cite{Deniz2022ISIT}. AirPooling is a task-oriented AirComp technique targeting AIoT sensing. By designing AirPooling, we aim to address the following two open issues.
\begin{itemize}
\item The first issue is how to realize Max-Pooling using AirComp. The class of
AirComputable functions are termed \emph{nomographic functions} characterized by a summation form with different pre-processing of summation terms and post-processing of the summation~\cite{GX2019IOTJ}. Examples include averaging and geometric mean. Nevertheless, the maximum function underpinning the Max-Pooling is not a nomographic function and thus does not allow direct AirComp implementation. 
\item The second issue is how to design AirPooling targeting a specific computation task with an associated  E2E performance metric. In particular, considering the task of classification, which is common in AIoT sensing for object recognition, the task-oriented design of AirPooling should aim to improve the classification accuracy in the presence of channel hostility. Existing AirComp techniques lack E2E awareness as they have been  designed largely using the generic metric of \emph{mean squared error} (MSE) w.r.t. to the \emph{noiseless} case~{\cite{GX2019IOTJ}}. This makes AirPooling, the theme of this work, an uncharted area.
\end{itemize}

The key contributions and findings of the work are summarized as follows.
\begin{itemize}
    \item \textbf{Designing Generalized AirPooling}: We propose the novel technique of generalized AirPooling that includes Average-AirPooling and Max-AirPooling as two special cases. 
    The design leverages the following properties of of the generalized $p$-norm:
    \begin{equation}
        \Vert x\Vert_p=\left(\sum_{n=1}^N |x_n|^p\right)^{\frac{1}{p}}  \begin{cases}
            = \sum_{n=1}^{N} |x_n|, & p = 1, \\
            \rightarrow \max\limits_{n} |x_n|, & p \rightarrow \infty.
        \end{cases}
    \end{equation}
    It features an air-interface function  controlled by a so-called  \emph{configuration parameter}  (i.e., $p$)  such that the function  implements Average-AirPooling when the parameter is equal to one and approaches Max-AirPooling when the parameter grows. Then generalized AirPooling can be realized by decomposing the air-interface function into \emph{pre-processing} at devices and \emph{post-processing} at the server on top of the conventional AirComp.
    
    \item \textbf{E2E Performance Analysis of AirPooling:} Consider the popular E2E task of classification for object and pattern recognition. Though the direct analysis of E2E classification accuracy is intractable, we overcome the difficulty by proposing an indirect approach for which the classification margin theory is used to relate the accuracy loss and AirPooling error as induced by channel distortion. The result allows the AirPooling error to be used as a tractable surrogate of classification accuracy loss in subsequent analysis and designs. In particular, we derive an \emph{approximation-noise} tradeoff between the Max-AirPooling approximation accuracy and noise amplification, which is regulated by the configuration parameter. A similar tradeoff does not exist for Average-AirPooling that can be implemented directly with AirComp without additional pre-processing as Max-AirPooling, implying setting the parameter equal to one. 
    
    \item \textbf{Optimization of AirPooling:} The configuration parameter of AirPooling is optimized under the criterion of  minimum AirPooling error for given receive signal power. For tractability, the error is approximated by a  derived  upper bound. Then leveraging the preceding approximation-noise tradeoff, the near-optimal parameter  is derived in closed-form for Max-AirPooling. The result  favours a large parameter in the case of high transmit power to better approximate the pooling function but a small parameter   in the case of low transmit power to avoid noise amplification. On the other hand, we show that setting the parameter equal to one is optimal for Average-AirPooling regardless of the SNR. 

    \item \textbf{Experiments:} Experimental results using selected real-life datasets demonstrate that Max-AirPooling can achieve  higher  sensing accuracies than Average-AirPooling, supporting the need of developing generalized AirPooling. Furthermore, compared with a digital air interface, AirPooling is shown to reduce air-latency by orders-of-magnitude while achieving comparable sensing accuracies.  
   
\end{itemize}

The remainder of this paper is organized as follows. The sensing and communication models are introduced in Section~\ref{sec: system_model}. The principle design  of generalized AirPooling is presented in Section~\ref{sec: airpooling_preliminaries}. Performance analysis comprising E2E performance metrics and tradeoffs is developed in Section~\ref{Sec: Performance-transalation} while the configuration parameters are optimized for Max-AirPooling and Average-AirPooling in Section~\ref{sec: optimization}. 
Section~\ref{sec: experiments} reports the numerical evaluation of AirPooling, followed by concluding remarks in Section~\ref{sec: conclusion}. 

\section{System Model}
\label{sec: system_model}
As illustrated in Fig.~\ref{fig: system_diagram}, we consider an AIoT sensing system where $K$ sensors, wirelessly connected to an edge server, cooperate to complete an inference task (e.g., object recognition). In a sensing round, each sensor acquires its view (e.g., an image) of a common object from a particular perspective and extracts features using the sensor model. The server then aggregates the features uploaded by sensors to infer the classification label of the object. Models and metrics are described in sub-sections.

\subsection{Distributed Sensing Model}
\label{subsec: sensing model}
The distributed sensing system is based on the mentioned MVCNN architecture. A pre-trained {sensor model} is deployed on each sensor, say sensor $k$, that takes the captured image as input and outputs a feature map comprising $N$ real features, denoted in its vectorized form as $\bff_k \in \mathbb{R}^N$. We consider the popular deep learning architecture where features are  outputs of the non-negative \emph{ReLU} or \emph{sigmoid} functions, and thus $f_{k,n}\geq 0$ for $n=1,\ldots,N$. At the server, the local feature maps $\{\bff_k\}_{k=1}^K$ undergo the view-pooling operation, i.e., being aggregated into a single feature map $\bg \in \mathbb{R}^N$ before $\bg$ is fed into the pre-trained {server model} to obtain the inference result. Specifically, 
the feature pooling is materialized via \emph{average-pooling} or \emph{max-pooling} with the $n$-th pooled feature given as
\begin{equation}\label{perfect-pooling}
    g_{{\sf avg}, n}  \triangleq\frac{1}{K}\sum_{k=1}^{K} f_{k,n} , \quad
    g_{{\sf max}, n}  \triangleq\max\limits_{k} f_{k,n}.
\end{equation}

The AirPooling implementation of the above operations is designed in the next section, while some needed notation and metrics are defined as follows. Define $\hat{\mathbf{g}}\triangleq\left[\hat{g}_1, \cdots, \hat{g}_N\right]$ where $\hat{g}_n$ denotes the $n$-th AirPooled feature. To quantize its channel distortion and the resultant effect on sensing performance, we introduce two metrics. First, the \emph{AirPooling error}, ${D}_n$, of an AirPooled feature $\hat{g}_n$ is defined using the MSE of that feature against its ground truth $g_n$, $D_n \triangleq \mathbb{E}\left[(\hat{g}_n - g_n )^2\right]$,
where $g_n$ is either $g_{{\sf avg}, n}$ or $g_{{\sf max}, n}$ as appropriate. Second, the \emph{classification accuracy} (i.e., the rate of correct recognition) at the server inference model, into which AirPooled features are fed, is defined as $R_{\sf AP} \triangleq \mathbb{E}\left[\mathcal{I}\left(\hat{\ell}=\ell\right)\right]$ with the expectation taken over both channel noises and data samples, where $\mathcal{I}(\cdot)$ denotes the indicator function, $\hat{\ell}$ the predicted label, and $\ell$ the ground-truth label of a data sample, respectively.

\begin{figure}
    \centering
    \includegraphics[width=0.99\textwidth]{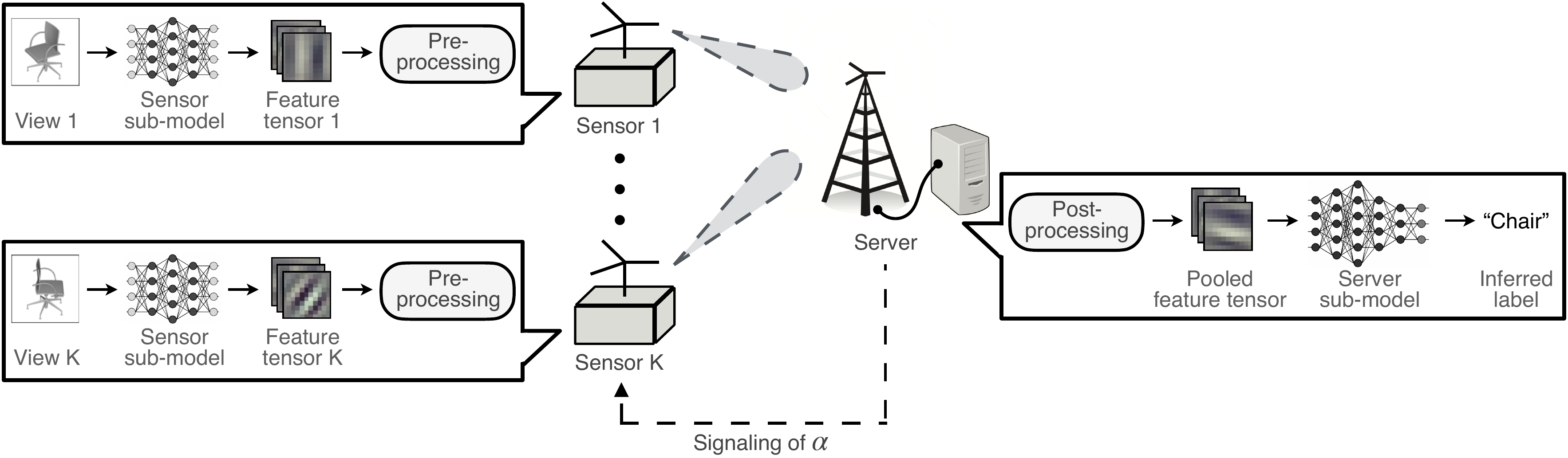}
    \caption{Distributed sensing system with AirPooling.}
    \label{fig: system_diagram}
\end{figure}

\subsection{Communication Model}
\label{sec: commun_model}
In the considered system in Fig.~\ref{fig: system_diagram}, the server and sensors are equipped with single antennas. The acquisition of the $N$ pooled features, i.e., $\{g_n\}$, is conducted sequentially using $N$ symbol durations. We assume synchronization between sensors at the symbol level. Consider the $n$-th symbol duration, where $K$ local features, $\{f_{k,n}\}_{k=1}^{K}$, will be aggregated to obtain $g_n$. The channel gains between sensors and the server 
are assumed known at both sides. We further assume \emph{independent and identically distributed} (i.i.d.)  block fading channels. For ease of notation, the index $n$ is omitted in the following expositions, which are valid for an arbitrary index. To realize AirPooling, all the sensors simultaneously transmit single data symbols, which results from pre-processing of local features (elaborated in Section~\ref{subsec: generalized_airpooling}). The aggregated symbol is given by
\begin{equation}
    \label{eqn: symbol-aggregation}
    y=\sum_{k=1}^K h_{k} p_{k} s_{k}+z,
\end{equation}
where $h_{k}$ is the channel gain between the $k$-th sensor and server, $s_{k}$ the symbol transmitted by the $k$-th sensor, $p_{k}$ the precoding coefficient, and $z\sim \mathcal{C}\mathcal{N}(0,\sigma^2)$ the additive white Gaussian noise with power $\sigma^2$, respectively. Channel inversion precoding is adopted as given by $p_k=\frac{\sqrt{P_\mathsf{rx}}}{h_k}$,
where $P_\mathsf{rx}$ denotes the receive power level coordinated by the server~\cite{GX2020TWC}.
In the case of deep fade, the scheme of truncated channel inversion can be used instead to avoid excessive transmit power consumption~\cite{GX2020TWC}.
The transmit power constraint for sensor $k$ is given by $\mathbb{E}[|p_k s_k|^2] \leq P_0$, where $P_0$ is the power budget at each sensor.

\section{Over-the-Air Multi-View Pooling}
\label{sec: airpooling_preliminaries}
\subsection{Generalized AirPooling}
\label{subsec: generalized_airpooling}

The proposed generalized AirPooling builds on the AirComp operation in~\eqref{eqn: symbol-aggregation} to include additional {pre- and post-processing}, thereby enabling reconfigurability to realize Average- or Max-AirPooling. The details as well as the properties of generalized AirPooling are provided as follows.

At each sensor, say sensor $k$, the local feature $f_k$ is pre-processed with a power function to generate the pre-processed feature value $v_k$, as given by $v_k =  f_k^{\alpha}$, 
where $\alpha$ denotes the tunable configuration parameter. Following the literature~\cite{GX2020TWC,Xiaowen2020TWC},
to facilitate transmit power control, the transmitted symbol of sensor $k$ is 
modulated via normalizing $v_k$: $ s_k =
    \frac{1}{\nu_\alpha}(v_k - \eta_\alpha),$
where the normalization parameters $\eta_\alpha\triangleq\mathbb{E}[v_k]$ and $\nu_\alpha^2\triangleq\mathbb{E}[(v_k-\eta_\alpha)^2]$ depend on both $\alpha$ and the underlying feature distribution.
The received symbol at the edge server is given by
\begin{equation}
    y=\sum_{k=1}^K \sqrt{P_\mathsf{rx}} s_k + z,
\end{equation}
Under sensors' transmit power constraints, the received power level is constrained by $P_\mathsf{rx} \leq \bar{P}$, where $\bar{P}\triangleq\frac{P_0 }{\mathbb{E}\left[ |h_k|^{-2}\right]}$. 
The aggregated feature before post-processing, $\hat{v}$, is then obtained by de-normalization
\begin{align}
    \label{eqn: denormalization}
    \hat{v}  =\frac{{\nu}_\alpha}{\sqrt{P_\mathsf{rx}}} y + \eta_\alpha K  = \sum_{k=1}^K v_k + \xi,
\end{align}
where $\xi\triangleq\nu_\alpha z/{\sqrt{P_\mathsf{rx}}}$ is the equivalent aggregation noise which is zero-mean Gaussian distributed with variance $\sigma_\xi^2=\sigma^2\nu_\alpha^2/P_{\sf rx}$.
Last, the server post-processes $\hat{v}$ to estimate the pooled feature value as given by
\begin{equation}\label{post_processing}
    \hat{g} = \left[\left(\frac{\hat{v}}{\beta}\right)^+\right]^{1/{\alpha}} = \left[\frac{1}{\beta}\left(\sum_{k=1}^K  f^{\alpha}_k+\xi\right)^+\right]^{1/{\alpha}},
\end{equation}
where the \emph{ramp function} $(\cdot)^+$ is defined as $(\cdot)^+\triangleq \max{\{\cdot,0\}}$, and $\beta$ is a tunable parameter termed the \emph{post-processing parameter}. {Here the ramp function reduces the AirPooling error by zeroing negative observations since all feature values are non-negative.} To facilitate performance analysis in the sequel, it is useful to introduce the noise-free version of $\hat{g}$ ($\xi=0$), denoted as $\tilde{g}$ and given as
$ \tilde{g} = \left(\frac{1}{\beta}\sum_{k=1}^K f^{\alpha}_k\right)^{1/{\alpha}}.$

Next, we prove the reconfigurability of the generalized AirPooling in~\eqref{post_processing}. To this end, the optimal setting of the parameter, $\beta$, in~\eqref{post_processing}, is characterized as follows.
\begin{lemma} 
    \label{lemma: post_processing_factor}
    \emph{When the channel noise is negligible, to minimize the AirPooling error, the optimal post-processing parameter, $\beta^*$, should be set for Average-AirPooling as $\beta^{*}=K$; and for Max-AirPooling as}
    \begin{equation} \label{beta_n}
        \beta^* = \left(\frac{\mathbb{E}\left[\Vert\mathbf{f}\Vert_{\alpha}^2\right]}{\mathbb{E}\left[f_{\max}\Vert\mathbf{f}\Vert_{\alpha}\right]}\right)^{\alpha},
    \end{equation}
    \emph{where $\mathbf{f}\triangleq[f_1,f_2,\ldots,f_K]$, $f_{\max}\triangleq\max{\{f_k\}}$, and $\Vert\cdot\Vert_{\alpha}$ denotes the $\ell_\alpha$-norm. Moreover, $\beta^*$ is bounded as $1\leq \beta^*\leq K$. }
\end{lemma}
\begin{proof}
    See Appendix~\ref{proof: best_beta}.
\end{proof}
Using this result, the said reconfigurability of the generalized AirPooling can be proved below.
\begin{theorem}
    \label{theorem: approximation_capability}
    \emph{When channel noise is negligible ($\xi=0$), AirPooling is capable of errorless implementation of either Average-Pooling and Max-Pooling:}
    \begin{equation}
        \label{eqn: approximation_capability}
         \begin{cases}
    \tilde{g} \rightarrow g_{\sf max}, & \beta = \beta^*,\quad \alpha\rightarrow\infty, \\
    \tilde{g} = g_{\mathsf{avg}},  & \beta=K,\quad \alpha=1,
        \end{cases}
    \end{equation}
\emph{where $\tilde{g}$ is given as $ \tilde{g} = \left(\frac{1}{\beta}\sum_{k=1}^K f^{\alpha}_k\right)^{1/{\alpha}}$ and $\beta^*$ is given in Lemma~\ref{lemma: post_processing_factor}.}  
\end{theorem}
\begin{proof} 
    See Appendix~\ref{app: proof_approximation_capability}.
\end{proof}
The above result shows generalized AirPooling's versatility in accurate over-the-air functional approximation when channel noise is negligible. But in the presence of noise, its configuration parameter needs to readjusted to balance functional approximation and reining in the noise effect. This is the main to address in the following sections.
Last, the generalized AirPooling protocol is summarized in Algorithm~\ref{algorithm: airpooling}.

\begin{algorithm}[t]
\caption{AirPooling Protocol}
\label{algorithm: airpooling}
\textbf{Input:} $P_0$, $K$, $\mathbb{E}[g^2_{\max}|K]$\;
\textbf{(Configuration)} The AP determines and broadcasts the configuration parameter $\alpha$ and the coordinated receive power level $P_{\sf rx}$\;
\textbf{for} feature dimension $n=1,2,\cdots,N$ \textbf{do}\\
    \begin{enumerate}
        \item[1:]\!\!\!\! {(Sensors transmission)} Each sensor (sensor $k$) exploits the {function} $v_{k,n} = f_{k,n}^{\alpha}$ to pre-process its local feature $f_{k,n}$ and simultaneously transmit the normalized symbols;
        \item[2:]\!\!\!\! {(Post-processing at the server)} The server invokes the post-processing function~\eqref{post_processing} to obtain an estimation of pooled feature, $\hat{g}_n$;
        
    \end{enumerate}
    
\textbf{end for} \\    
\textbf{Output:} $\hat{g}_1, \hat{g}_2, \cdots, \hat{g}_{N}$.
\end{algorithm}

\subsection{Other AirPooling Functions}
Besides Max- and Average-Pooling, AirPooling can be extended to a number of other functions described as follows.
\begin{itemize}
    \item \emph{Weighted-sum Pooling:} Sensors scale local features with weights broadcast from the server, then execute Average-AirPooling of weighted features.
    \item \emph{Concatenation Pooling:} For a server inference model whose layer right after concatenation is a linear one such as dense and convolution layers, concatenation can be recast to average-pooling~\cite{FU2019Neurocomp}. The recasting essentially relies on relocating different branches of neurons in that linear layer to corresponding sensors. 
    \item \emph{Average-Max Hybrid Pooling:} Varying the configuration parameter $\alpha$ offers the opportunity of an in-between point transiting between two extreme cases $\alpha=1$ and $\alpha \rightarrow {\infty}$, e.g., square-root pooling ($\alpha=2$) meant for recognition problems in~\cite{Thomas2009CVPR}. Interestingly, such hybrid pooling is reported to boost sensing performance for some feature extractors pre-trained even without pooling functions~\cite{LeCun2010ICML}.
\end{itemize}

\section{Classification Accuracy and AirPooling Error}
\label{Sec: Performance-transalation}
In this section, we analyze the relationship between the two performance metrics, namely, classification accuracy and AirPooling error, so that the latter can be used as a surrogate of E2E performance metric to allow tractable optimization of AirPooling. Using the derived bound, we characterize the fundamental tradeoffs in AirPooling.

\subsection{Relationship between Classification Accuracy and AirPooling Error}
Given a trained MVCNN model with noiseless multi-view pooling, its classification accuracy, $R_0$, is the probability that an input sample is correctly classified. To quantify the metric 
we leverage the concept of  \emph{classification margin}~\cite{Jure2017TSP}. Let $d(\cdot,\cdot)$ denote the Euclidean distance in the feature space $\mathbb{R}^N$. The (trained) server model has an intrinsic classification margin $\Delta$ in the feature space, which refers to the infimum of the distance from an arbitrary perfectly pooled feature vector, $\mathbf{g}$, to the classification boundary. Consequently, if $\bg$ is correctly classified, then $d(\bg,\hat{\bg})<\Delta$ is a sufficient condition for correct classification of the perturbed feature vector $\hat{\bg}$ obtained by AirPooling.

\begin{lemma}\label{lemma: acc_lb}
\emph{The classification accuracy of the server model with AirPooling, denoted as $R_{\sf AP}$, is lower bounded as}
\begin{eqnarray}
     R_{\sf AP} 
     &{\geq}&  R_0 \mathsf{Pr}[\Vert\bee\Vert_2 <\Delta] \label{eqn: accuracy_lowerbound_0}\\ 
     & {\geq}& R_0\left(1-\frac{ {D}_{\Sigma} }{\Delta^2}\right)  \triangleq  R_{\sf AP}^*,
     \label{eqn: accuracy_lowerbound}
\end{eqnarray}
\emph{where $\bee=\hat{\bg}-\bg$ is the error vector and  ${D}_{\Sigma}\triangleq\sum_{n=1}^{N} D_n$ is the sum of AirPooling errors over all $N$ feature dimensions. }
\end{lemma}
\begin{proof}
    See Appendix~\ref{proof: acc_lb}.
\end{proof}

The two inequalities in Lemma~\ref{lemma: acc_lb} suggest two tractable methods for achieving a target classification accuracy via analyzing the AirPooling error. First, using the inequality in~\eqref{eqn: accuracy_lowerbound_0}, a sufficient condition for achieving a target accuracy, denoted as $R_{\sf target}$ is 
\begin{equation}
    \label{eqn: sufficient_cond_1}
    \mathsf{Pr}\left( \Vert \mathbf{e} \Vert_2 < \Delta \right)\geq \frac{R_{\sf target}}{R_0} . 
\end{equation}
Second, using the inequality in~\eqref{eqn: accuracy_lowerbound}, the other sufficient condition is  
\begin{equation}
    \label{eqn: sufficient_cond_2}
    {D}_{\Sigma} \leq \Delta^2 \left( 1 - \frac{R_{\sf target}}{R_0} \right),
\end{equation}
that is simpler than the preceding one but looser.

\begin{Example}
\emph{Consider Average-AirPooling over $K$ sensors with $\beta=K$ , $\alpha=1$, and $\hat{g}_n=g_n+\frac{\xi_n}{K}$ for example. Thus we have $e_n=\frac{\xi_n}{K}$. The distance between $\hat{\bg}$ and $\bg$ is given by $\Vert\bg-\hat{\bg}\Vert_2= \Vert\bee\Vert_2$ which follows a $\chi$-distribution with $N$ degrees of freedom. Invoking the {cumulative distribution function} of $\chi$-distributions and the inequality in~\eqref{eqn: accuracy_lowerbound_0}, $R_{\sf AP} \geq  R_0 P(\frac{N}{2},\frac{K^2\Delta^2}{2\sigma_\xi^2})=R_0 P(\frac{N}{2},\frac{N\Delta^2}{2{D}_{\Sigma}})$,
where $P(k,x)=\frac{\gamma(k,x)}{\Gamma(k)}$ is the regularized gamma function with $\gamma(k,x)=\int_0^x t^{k-1}e^{-t}dt$ denoting the lower incomplete gamma function and $\Gamma(k)=\int_0^\infty t^{k-1}e^{-t}dt$ denoting the gamma function. Since $P(k,x)$ is a monotone increasing function of $x$, the sufficient condition in~\eqref{eqn: sufficient_cond_1} is particularized for the current case as ${D}_{\Sigma} \leq \frac{N\Delta^2}{2P^{-1}\left(N/2,{R_{\sf target}}/{R_0}\right)}$, where $P^{-1}(k,x)$ is the inverse function of $P(k,x)$.}
\end{Example}

For the generalized AirPooling comprising Max-AirPooling, its error is obscured by two confounding effects, namely channel noise perturbation and the function-approximation error. This ambiguity renders the characterization of $D_n$ complicated and motivates us to overcome this difficulty in the sequel.

\subsection{Bounding AirPooling Error} 
In the previous sub-section, the classification accuracy is related to the AirPooling error. Here, we analyze the error by decomposing it into components associated with different error types. Consider $D_n$ for an arbitrary $n$ and define it as a function of $\alpha$, $D_n=D(\alpha)$. 
To this end, we define two useful functions $D_{\sf chan}(\alpha)\triangleq \mathbb{E}[|\hat{g}-\tilde{g}|^2]$ and $D_{\sf appr}(\alpha)\triangleq \mathbb{E}[|\tilde{g}-g|^2]$, where $\tilde{g}$ is the pooled feature assuming negligible channel noise. The first function, $D_{\sf chan}(\alpha)$, accounts for the perturbation due to channel noise, termed the \emph{noise-perturbation error}. The other function, $D_{\sf appr}(\alpha)$, represents the \emph{function-approximation error} in approximating the max-operator in~\eqref{perfect-pooling} by summation with pre- and post-processing (see~\eqref{post_processing}).
\begin{lemma}
\label{proposition_decompose}
\emph{The AirPooling error can be upper bounded as
\begin{equation}
    \label{eqn: decompose_airpoolingerror}
    D(\alpha)\leq c_0 \left[ D_{\sf chan}(\alpha) + D_{\sf appr}(\alpha)\right], 
\end{equation}
where $c_0=1$ for Average-Pooling and $c_0=2$ for Max-Pooling. }
\label{prop: mse_decomposition}
\end{lemma}
\begin{proof}
See Appendix~\ref{app: proof_proposition}.
\end{proof}

The characterization of error components, $D_{\sf chan}(\alpha)$ and $D_{\sf appr}(\alpha)$, can be made simple via the analysis of the mean and variance over pooling i.i.d. random features following the assumption in~\cite{LeCun2010ICML}. This yields the following result.
\begin{lemma}
\label{lemma: mse_max_no_selection}
\emph{The noise-perturbation error and the function-approximation error are upper bounded as
\begin{align}
    \label{eqn: define_delta}
    D_{\sf chan}(\alpha)&\leq \left(\frac{\sigma^2\nu_{\alpha}^2}{P_\mathsf{rx}}\right)^{\frac{1}{\alpha}} \triangleq \delta, \\ 
    D_{\sf appr}(\alpha)&\leq \epsilon = \begin{cases}
        \left(1-K^{-\frac{1}{\alpha}}\right)\mathbb{E}[f_{\sf max}^2 \vert K ]\triangleq \epsilon_{\sf m}, & \text{Max-Pooling as the ground truth,} \\
        \mathbb{E}\left[(\frac{1}{K}\Vert\bff\Vert_\alpha-g_\mathsf{avg})^2\right]\triangleq\epsilon_{\sf a}, & \text{Average-Pooling as the ground truth}.
    \end{cases}
    \label{eqn: define_epsilon}
\end{align}
}
\end{lemma}
\begin{proof}
See Appendix~\ref{app: proof_lemma_mse_max_noselection}. 
\end{proof}
These results help the understanding of some useful tradeoffs in AirPooling discussed in the sequel.

\subsection{Tradeoffs in AirPooling}
\label{subsec: tradeoffs-in-airpooling}
First, consider Max-AirPooling. There exists a tradeoff between functional approximation and channel-noise suppression. On one hand, it can be observed from~\eqref{eqn: define_epsilon} that the function-approximation error, $\epsilon_\mathsf{m}$ monotonically decreases as the configuration parameter $\alpha$ grows. Particularly, $\lim_{\alpha \rightarrow \infty} \epsilon_m=0$ and correspondingly Max-AirPooling is asymptotically achieved as shown in Theorem~\ref{theorem: approximation_capability}. On the other hand, a larger $\alpha$ tends to amplify channel noise as explained below. Considering the pre-processing function (see Section~\ref{subsec: generalized_airpooling}), increasing $\alpha$ makes the transmitted features, $\{f_k^\alpha\}$, highly skewed. Specifically, those features with small magnitudes are suppressed and their submission is prone to channel distortion. As a result, the noise-perturbation error bound in \eqref{eqn: define_delta}, $\delta=\left(\frac{\sigma^2\nu_{\alpha}^2}{P_\mathsf{rx}}\right)^{\frac{1}{\alpha}}$, tends to grow as $\alpha$ increases. For example, this term is shown to increase asymptotically linearly with $\alpha$ in the case of rectified Gaussian features analyzed in Section~\ref{subsec: optimize_alpha}.  The above tradeoff between \emph{function approximation} and \emph{channel noise mitigation} is demonstrated numerically in Fig.~\ref{fig: tradeoff_1}. This tradeoff necessitates the optimization of configuration parameter $\alpha$, which is addressed in the next section. 
\begin{figure*}[t]
\centering
\subfigure[Uniform distribution]{\includegraphics[height=4.1cm]{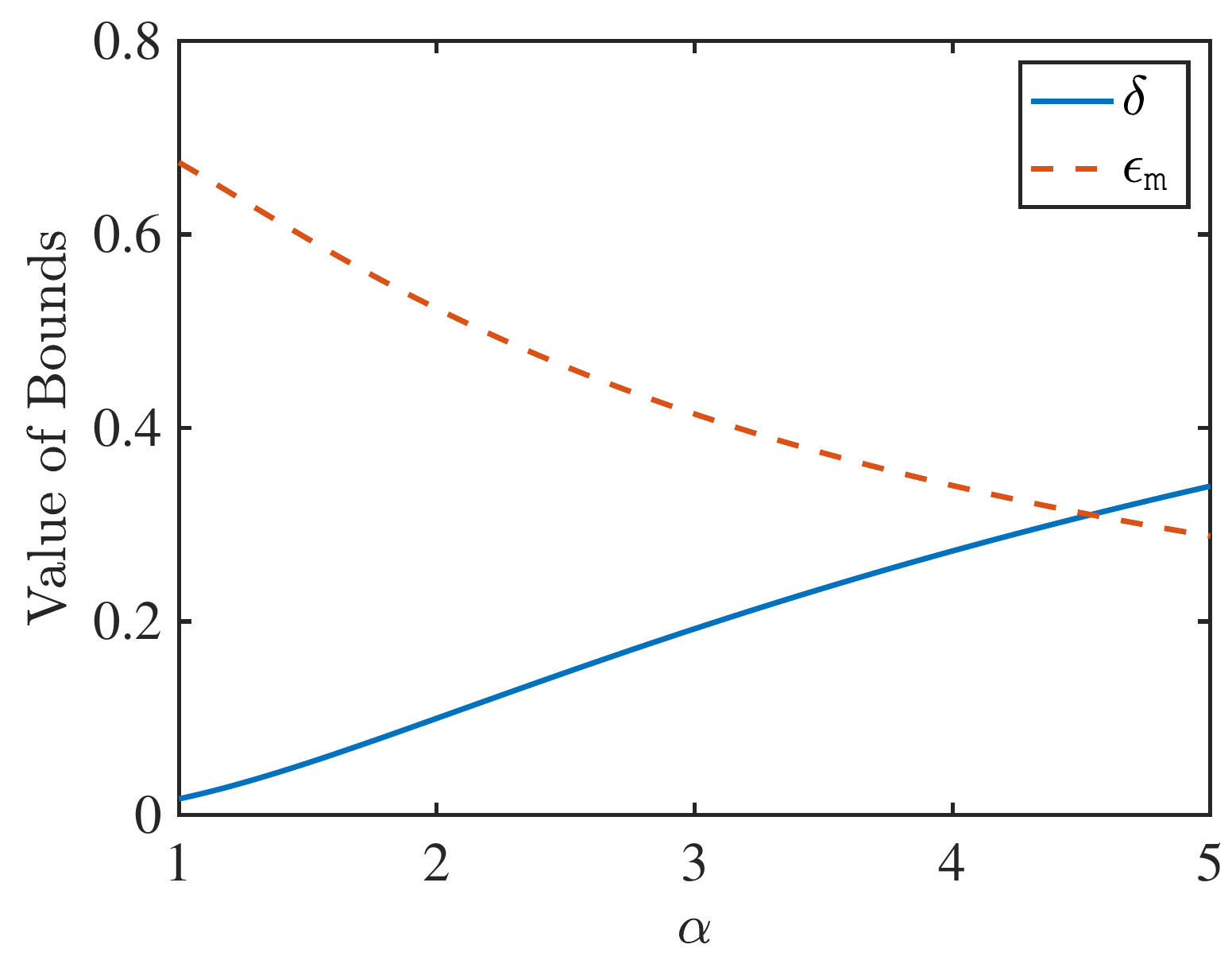}}
\hspace{0.05cm}
\subfigure[Rectified Gaussian distribution]{\includegraphics[height=4.1cm]{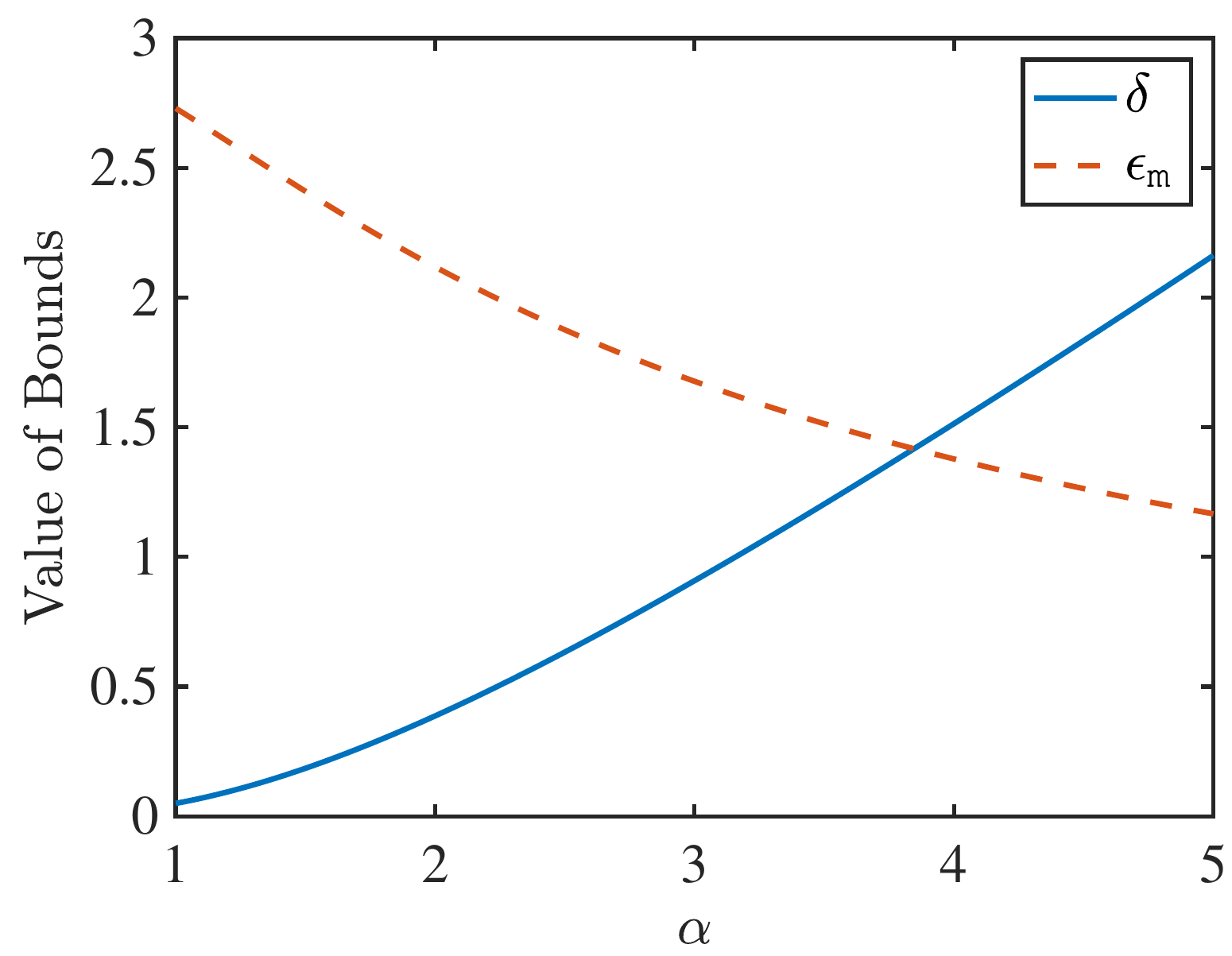}}
\hspace{0.05cm}
\subfigure[Exponential distribution]{\includegraphics[height=4.1cm]{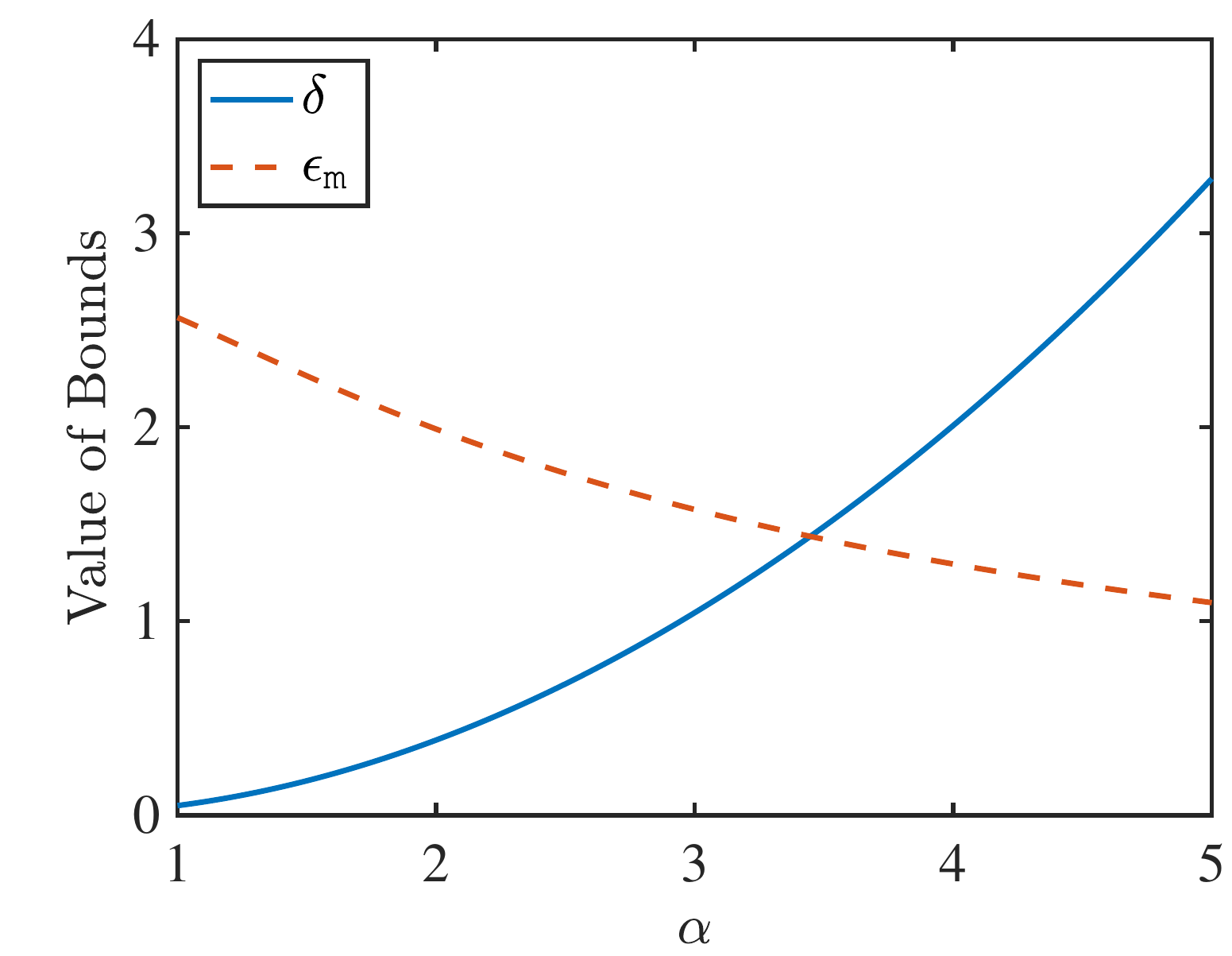}}
\caption{Numerical evaluations of the tradeoff between the noise-perturbation error ${\delta}$ and the function-approximation error ${\epsilon}_{\sf m}$ as controlled by $\alpha$ over types of three feature distributions.}
\label{fig: tradeoff_1}
\end{figure*}

Next, for Average-AirPooling, the simple setting of $\alpha=1$ nulls the function-approximation error. The design challenge centers on quantifying the behavior of the noise-perturbation error, $\delta$, against the configuration parameter $\alpha$ so as to reduce AirPooling error. More details are given in the sequel.

\section{Optimization of AirPooling}
\label{sec: optimization}
In the preceding section, we established the AirPooling error as a surrogate of classification accuracy and derived the inherent tradeoff in AirPooling. In this section, we use the results to minimize the AirPooling error by controlling the configuration parameter as well as coordinating the receive signal power. 

\subsection{Problem Formulation}
\label{subsec: pronblem_formulation}
For either Max- or Average-AirPooling, the AirPooling-error minimization problem is formulated as
\begin{equation*}\text{(P1)}\quad 
\begin{aligned}
\min\limits_{\alpha, P_{\sf rx}}\quad &\, {\delta}+{\epsilon} \\
\mathrm{s. t.}\quad &  \alpha \geq 1,\\
                    &  P_{\sf rx} \leq \bar{P}, 
\end{aligned}
\end{equation*}
where the AirPooling error components are defined in Lemma~\ref{lemma: mse_max_no_selection}. It can be observed from~\eqref{eqn: define_delta} and~\eqref{eqn: define_epsilon} that ${\delta}$ is monotonically decreasing with $P_\mathsf{rx}$ while $\epsilon$ is independent of $P_\mathsf{rx}$. Hence, the optimal value of $P_\mathsf{rx}$ is given as $P^*_\mathsf{rx} = \bar{P}$. As a result, ${\delta}=\left(\frac{\sigma^2}{\bar{P}}\nu^2_{\alpha}\right)^{\frac{1}{\alpha}}$ and Problem (P1) reduces to
\begin{equation*}\text{(P2)}\quad 
\begin{aligned}
\min\limits_{\alpha}\quad &\, \left(\frac{\sigma^2}{\bar{P}}\nu^2_{\alpha}\right)^{\frac{1}{\alpha}}+{\epsilon} \\
\mathrm{s. t.}\quad &  \alpha \geq 1. 
\end{aligned}
\end{equation*}

\subsection{Optimization of Max-AirPooling}
\label{subsec: optimize_alpha}
To solve Problem (P2), we adopt an assumption on the feature distribution following the learning literature (see, e.g., \cite{Socci1997NeurIPS,Sun2021NeurIPS}), that results from Gaussian distributed raw features  traversing through a ReLU activation unit.
\begin{assumption}[\emph{Rectified Gaussian Distribution}]
\label{assume: rectified_gaussian}
\emph{The features are  given by $f_k = \max\{\tilde{f}_k,0\},$}
\emph{where the raw features $\{\tilde{f}_k\}$ are i.i.d. Gaussian distributed with $\tilde{f}_k \sim \mathcal{N}(0,1)$. Mathematically, 
\begin{equation}
    \label{eqn: rectified_gaussian_pdf}
    \Pr\left(f_k=x\right) = 
    \begin{cases}
        \frac{1}{2}, & x=0, \\
        \frac{1}{\sqrt{2\pi}}e^{-\frac{x^2}{2}}, & x > 0.
    \end{cases}
\end{equation}}
\end{assumption}
\noindent Using Assumption~1, we can approximately solve Problem (P2) in closed form though its exact solution is difficult to find. The key is to use the fact that $\alpha$ tends to be large for Max-AirPooling (see Theorem~\ref{theorem: approximation_capability}) to make several approximations.

First, the noise-perturbation error, $\delta$, in the objective of Problem (P2) can be approximated by a simpler expression. To this end, the variance of features raised to the power of $\alpha$, 
$\nu_{\alpha}^2$, can be written as
\begin{align}
    \nu^2_{\alpha}&=\mathbb{E}[(|f_k|^{\alpha}-\mathbb{E}[|f_k|^{\alpha}])^2] = \mathbb{E}[|f_k|^{2\alpha}]-\mathbb{E}^2[|f_k|^{\alpha}].
\end{align}
Based on the distribution in~\eqref{eqn: rectified_gaussian_pdf}, 
\begin{eqnarray}
    \mathbb{E}\left[|f_k|^{s} \right] 
    = \int_{0}^{\infty} x^{s} \Pr(f_k=x) \ \mathrm{d} x = \frac{1}{\sqrt{\pi}} 2^{\frac{s}{2}-1}\Gamma\left(\frac{s+1}{2}\right). \label{eqn: noselection_moments}
\end{eqnarray}
Substituting~\eqref{eqn: noselection_moments} into~\eqref{eqn: define_delta} yields
\begin{equation}
    {\delta}
    \label{eqn: delta_m_gamma}
    = \left\{\frac{\sigma^2}{\bar{P}\sqrt{\pi}}2^{\alpha-1}\left[\Gamma\left(\frac{2\alpha+1}{2}\right)-\frac{\Gamma^2\left(\frac{\alpha+1}{2}\right)}{2\sqrt{\pi}}\right]\right\}^{1/\alpha} 
\end{equation}
Furthermore, the asymptotic growth rate of ${\delta}$ is characterized by the following lemma.
\begin{lemma}\label{lemma: delta_m_rate}
    \emph{The asymptotic growth rate of ${\delta}$ is given by}
    \begin{eqnarray}
        {\delta}
        \label{eqn: delta_m_gamma_1}
        &=& 2\left[\frac{\sigma^2}{2\bar{P}\sqrt{\pi}}\Gamma\left(\frac{2\alpha+1}{2}\right)\right]^{1/\alpha} + O \left(\frac{1}{2^\alpha}\right)\\
        &=& 
        2e^{-1}\left(\frac{\sigma^2}{\sqrt{2}\bar{P}}\right)^{\frac{1}{\alpha}}\alpha 
        + O \left(\frac{1}{\alpha}\right),
        \label{eqn: delta_m_gamma_2}
    \end{eqnarray}
\end{lemma}
\begin{proof}
   See Appendix~\ref{proof_lemma_delta_m_asym}. 
\end{proof}
It follows that for large $\alpha$, $\delta \approx \hat{\delta}$ with $\hat{\delta}$ defined as
\begin{equation}
    \hat{\delta} =  2e^{-1}\left(\frac{\sigma^2}{\sqrt{2}\bar{P}}\right)^{\frac{1}{\alpha}}\alpha . \nonumber
\end{equation}
The approximation is increasingly tight as $\alpha$ grows. Accordingly, Problem (P2) can be approximated as 
\begin{equation*}\text{(P3)}\quad 
\begin{aligned}
\min\limits_{\alpha}\quad &\, 2e^{-1}\left(\frac{\sigma^2}{\sqrt{2}\bar{P}}\right)^{\frac{1}{\alpha}}\alpha+{\epsilon}_{\sf m} \\
\mathrm{s. t.}\quad &  \alpha \geq 1. 
\end{aligned}
\end{equation*}

\begin{figure*}[t]
\centering
\subfigure[]{\includegraphics[height=5.9cm]{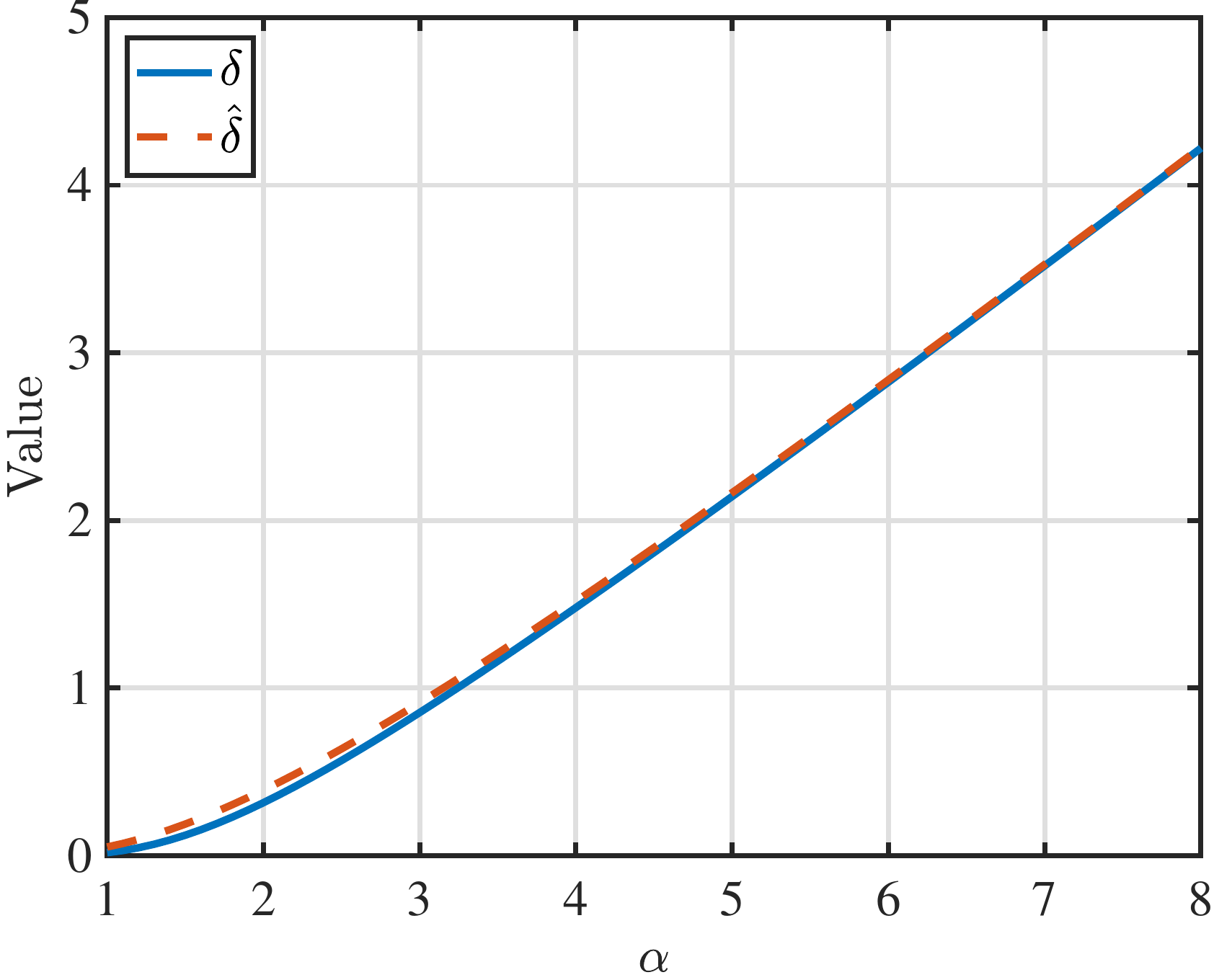}}
\hspace{0.5cm}
\subfigure[]{\includegraphics[height=5.9cm]{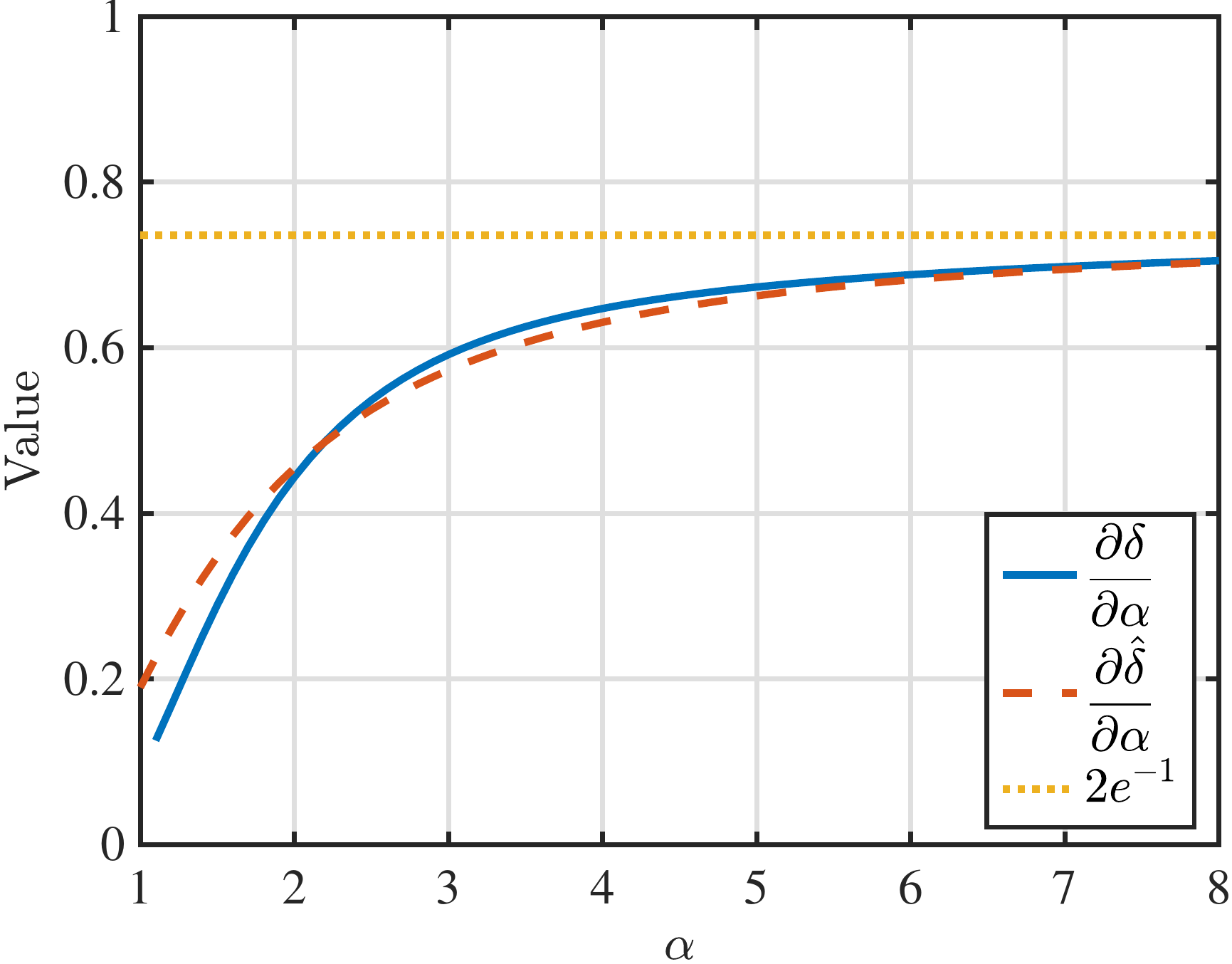}}
\caption{Numerical evaluation of (a) ${\delta}$, $\hat{\delta}$ and (b) their derivatives.}
\label{fig: epsilon1_results}
\end{figure*}

The key to solving Problem (P3) is to show that the two terms in its objective are monotone increasing and decreasing functions, respectively. These are the direct consequences of the tradeoff between functional approximation and noise mitigation as discussed in Section~\ref{subsec: tradeoffs-in-airpooling}. The desired results are obtained as shown in the following two lemmas.
\begin{lemma}
\label{lemma: asymp_gamma_difference}
\emph{ Given $\frac{\bar{P}}{\sigma^2}\geq 1$, the noise-perturbation error, {$\hat{\delta}(\alpha)$}, is a monotone increasing function since
\begin{equation}
    \frac{\partial \hat{\delta}}{\partial \alpha}=2e^{-1}\left(\frac{\sigma^2}{\sqrt{2}\bar{P}}\right)^{\frac{1}{\alpha}}\left[1+\frac{1}{\alpha}\log\left(\frac{\sqrt{2}\bar{P}}{\sigma^2}\right)\right] \geq 0.
\end{equation}
}
\end{lemma}
Note that for large $\alpha$, both ${\frac{\partial \delta}{\delta \alpha}}$ and ${\frac{\partial \hat{\delta}}{\delta \alpha}}$ are close to $2e^{-1}$. To illustrate $\hat{\delta}$ being an accurate approximation of $\delta_{\sf m}$ for large $\alpha$, 
we plot ${\delta}$, $\hat{\delta}$ and their derivatives against $\alpha$ in Fig.~\ref{fig: epsilon1_results}, where $\frac{\bar P}{\sigma^2}$ is set to be ${10}$. 
\begin{lemma}
    \label{lemma: monotone-epsilon}
    \emph{The functional approximation error, $\epsilon_{\sf m}(\alpha)$, is a monotone decreasing function.}
\end{lemma}
\begin{proof}
    The partial derivative of $\epsilon_{\sf m}(\alpha)$ is given by
    \begin{eqnarray}
    \label{epsilon_2_derivative}
    \frac{\partial {\epsilon}_{\sf m}}{\partial \alpha} = -\frac{1}{\alpha^2}K^{-\frac{1}{\alpha}}  \mathbb{E}[f_{\max}^2\vert K] \log{K} \leq 0. 
\end{eqnarray} 
\end{proof}
It follows from Lemmas~\ref{lemma: asymp_gamma_difference} and~\ref{lemma: monotone-epsilon} that the solution for Problem (P3) satisfies $\frac{\partial \hat{\delta}}{\partial \alpha}+\frac{\partial {\epsilon}_{\sf m}}{\partial \alpha}=0$, as specified by
\begin{equation}
    \label{eqn: noselection_suboptimal_equation}
    2e^{-1}\left(\frac{K\sigma^2}{\sqrt{2}\bar{P}}\right)^{\frac{1}{\alpha}}\left[1+\frac{1}{\alpha}\log\left(\frac{\sqrt{2}\bar{P}}{\sigma^2}\right)\right]=\frac{1}{\alpha^2}\mathbb{E}[f_{\max}^2\vert K]\log{K}. 
\end{equation}
Solving~\eqref{eqn: noselection_suboptimal_equation} yields the following main result of this section.
\begin{theorem}[\emph{Optimized Max-AirPooling}]\label{optimized_alpha_no_selection} \emph{Given $K\geq 4$ and $\bar{P}/\sigma^2>K$, the configuration parameter that approximately satisfies~\eqref{eqn: noselection_suboptimal_equation} and hence solves Problem (P2), denoted as  $\alpha^{*}$, is given as
\begin{equation}\label{eqn: optimized_alpha}
    \alpha^* = \frac{C}{W_0\left(\frac{2C(C+\log K)}{\exp{\left(1+\frac{C}{C+\log K}\right)} \mathbb{E}[f_{\max}^2\vert K]\log{K}}\right)+\frac{C}{C+\log K}},
\end{equation}
where $C \triangleq \log \left(\frac{\sqrt{2}\bar{P}}{K\sigma^2}\right)$ and $W_0(\cdot)$ denotes the principal branch of the Lambert $W$ function. }
\end{theorem}
\begin{proof}
    See Appendix~\ref{proof_optimal_alpha}.
\end{proof}
The configuration parameter, $\alpha^{*}$ in~\eqref{eqn: optimized_alpha}, balances Max-Pooling approximation and noise mitigation to maximize the classification accuracy. It can be observed from~\eqref{eqn: optimized_alpha} that $\alpha^{*}$ grows as the transmit power budget $\bar{P}$ increases. 
The reason is that as the noise effect diminishes, the functional approximation error dominates and can be suppressed by increasing $\alpha$ (see Theorem~\ref{theorem: approximation_capability}).

\begin{remark}
(Optimal Max-AirPooling in the Low-SNR Regime)
\emph{While Theorem 2 holds for the moderate- to high-SNR scenarios, the optimal configuration parameter in the low-SNR regime is discussed as follows. Let $\rho_0 \triangleq \frac{\sqrt{2}K}{e\mathbb{E}[f_{\max}^2\vert K]\log K}$. In the regime that the transmit power is no higher than the critical value $\rho_0$, it is easy to show that $\frac{\partial (\tilde{\delta}+{\epsilon}_{\sf m})}{\partial \alpha}>0$ at $\alpha=1$, indicating that suppressing channel noises now outweighs reducing the function-approximation error. In this case, $\alpha^*=1$, corresponding to Average-AirPooling. The reason is that the averaging operation of Average-AirPooling helps to suppress the channel noise such that its power is inversely proportional to the number of sensors, $K$.    }
\end{remark}

\subsection{Optimization of Average-AirPooling}

We are now in the position to optimize the control parameter for Average-AirPooling. Following the same methodology developed for Max-AirPooling, the optimized $\alpha$ can be easily derived for rectified Gaussian features.
\begin{proposition}
(Optimized Average-AirPooling).  \emph{Given $\bar{P}>\sigma^2$ and the feature distribution satisfying $(\nu_\alpha^2)^{\frac{1}{\alpha}}$ monotonically increasing in $\alpha$, the optimal configuration parameter that solves Problem (P2) is given by $\alpha^*=1$. }
\label{proposition: avg-airpooling}
\end{proposition}
\begin{proof}
   When $\alpha=1$, from \eqref{eqn: approximation_capability} we know that $\epsilon_{\mathsf{a}}=\mathbb{E}[(\tilde{g}-g_\mathsf{avg})^2|\alpha=1]=0$. On the other hand, given $P\geq\sigma^2$ and $(\nu_\alpha^2)^{\frac{1}{\alpha}}$ increasing in $\alpha$, $\delta = \left(\frac{\sigma^2\nu_{\alpha}^2}{P_\mathsf{rx}}\right)^{\frac{1}{\alpha}}$ is monotonically increasing with $\alpha$. Hence, $\alpha=1$ simultaneously minimizes $\delta$ and $\epsilon_{\mathsf{a}}$, and is thus optimal.
\end{proof}

\section{Experimental Results}
{In this section, we evaluate the proposed Max- and Average-AirPooling in a practical setting. We first present the experimental setup, and then study and compare the performance in various scenarios. }
\label{sec: experiments}

\subsection{Experimental Setup}

\subsubsection{System and Communication Settings}
Consider a distributed sensing system with $K=12$ cameras.
The total communication bandwidth between the cameras and the edge server  is $B=10$ MHz. It is divided into $M=12$ orthogonal sub-channels ({\emph{orthogonal frequency-division multiple access} (OFDMA)~\cite{Goldsmith2005}}), where each sub-channel is used for AirPooling of one feature dimension. In other words, each sensor transmits one symbol simultaneously with others over one sub-channel and $\frac{M}{B}$ seconds.  The receive noise power density at the server is set to $-174$ dBm with a hardware noise factor (also known as noise figure) $4$ dB. The noise power of each sub-channel is then $10^{-20}\frac{B}{M}$ Watt. The power budget per sub-channel is $\frac{P_0}{M}$, where $P_0$ is the total power budget of each sensor. We assume that the path-loss from devices to the server is all equal to $P_{\sf pl}=300^{-3.4}$ due to similar distances plus channel-inversion power control. Given the above settings, the  \emph{average receive SNR} ($\textrm{SNR}_{\textrm rx}$) of AirPooling equals to
\begin{equation}
    \label{eqn: experiment_rsnr}
    \textrm{SNR}_{\textrm rx} = \frac{P_0P_{\sf pl}}{ 10^{-20}  B  \mathbb{E}\left[\Vert h_k\Vert^{-2}\right] }, 
\end{equation}
where $\mathbb{E}\left[\Vert h_k\Vert^{-2}\right]$ is due to channel inversion. We consider \emph{Rician} fading  where the ratio of \emph{line-of-sight} (LoS) channel gain to non-LoS channel gain is set as $4$ dB. The latency (in seconds) of AirPooling is then given by $L_{\sf AirPooling} = \frac{N}{B}$,
where $N$ is the total number of features per sensor.

\begin{table}[t!]
\centering
\caption{Comparison of Classification Accuracies of Multi-View Max- and Average-Pooling for Distributed Sensing using the MVCNN Model and Real Datasets.}
\begin{tabular}{ccc}
\hline
Dataset & Max-Pooling & Average-Pooling  \\ \hline
ModelNet10       &         $91.67\%$         &     $89.17\%$                 \\
ModelNet40         &         $86.90\%$         &    $85.15\%$                  \\
ModelNet10-Shaded       &        $90.20\%$          &      $89.46\%$                \\
ModelNet40-Shaded       &       $87.33\%$           &      $86.52\%$                \\
ShapeNet       &      $94.18\%$        &      $93.45\%$                \\ \hline
\end{tabular}
\label{table: test_acc_modelnets}
\end{table}

\subsubsection{Model Architecture and Datasets} To implement the distributed sensing model, we adopt the mentioned MVCNN architecture with the 
with the \emph{VGG11} backbone for training and testing~\cite{Hang2015ICCV}. We build two variants of MVCNN: one for Max-Pooling and the other for Average-Pooling. We consider the celebrated \emph{ModelNet} and \emph{ShapeNet} datasets for multi-view object recognition. For each {ModelNet object class} (e.g., sofa), $12$ views per sample are captured by $K=12$ camera sensors with neighboring ones separated by an angle of $30^{\circ}$. There are two subsets of ModelNet, called ModelNet40 and ModelNet10, that comprise $40$ and $10$ most popular classes, respectively. By rendering views in a shaded style, ModelNet40 and ModelNet10 become ModelNet40-Shaded and ModelNet10-Shaded, generating another two subsets of ModelNet. On the other hand, ShapeNet consists of $13$ major classes (e.g., aeroplane and car) with 12 views per sample. Similar to ModelNet, the views are observed from uniformly separated angles. 

\subsubsection{Synthetic Dataset}
We construct a feature dataset for a multi-view binary classification task that is simple and requires only a \emph{shallow neural network} (SNN). Consider $K=4$ sensors. Each sample consists of $K=4$ views, where each view has $N=4$ features drawn from the rectified Gaussian distribution in~\eqref{eqn: rectified_gaussian_pdf}. The label of the sample is assigned according to its projection onto a pre-determined feature space via a non-linear transformation.  The SNN at the server, taking four pooled features as its input, comprises two \emph{fully-connected} (FC) layers with $5$ neurons each and a softmax output layer and is trained to classify the pooled feature vector. 

\subsubsection{Benchmarking Scheme}
We benchmark AirPooling against the traditional broadband digital multi-access scheme. It is termed \emph{digital air interface}. On each device, a single feature is quantized into $Q$ bits by uniform quantization. The server decodes the feature bit streams from sensors and pools them to obtain the pooled features. Varying $Q$ creates a tradeoff between the transmission latency and the fidelity, or inference accuracy.
To avoid multi-user interference, digital air interface also adopts OFDMA but each device uses only $\frac{M}{K}$ sub-channels instead of all $M$ sub-channels as for AirPooling. Each Rician fading sub-channel is inverted subject to the power constraint~\cite{Goldsmith2005}. The integer number of bits modulated into a symbol is adjusted according to the channel spectrum efficiency. The latency (in seconds) of the digital air interface scheme is then given by
    \begin{equation}
        \label{eqn: dpi_latency}
        L_{\sf digital} = \frac{KNQ}{B \log_2\left( 1 + \textrm{SNR}_{\textrm rx} K \right) },
    \end{equation}
    where $\textrm{SNR}_{\textrm rx}$ is given in~\eqref{eqn: experiment_rsnr}.

\subsection{Max-Pooling versus Average-Pooling}
We substantiate the need of designing Max-AirPooling using experimental results. 
To this end, putting the effect of wireless channels aside, the inherent superiority of Max-Pooling over Average-Pooling in distributed sensing is demonstrated based on the MVCNN model assuming reliable links and real datasets.
The performance comparison in terms of classification accuracy (averaged over three trials) is provided in Table~\ref{table: test_acc_modelnets}.
One can see that Max-AirPooling performs uniformly better than Average-Pooling over all considered datasets, which justifies the former’s popularity in practice.

\subsection{Parametric Optimization of Max-AirPooling}
\label{subsec: experiment_results}

\begin{figure*}[t]
\centering
\subfigure[AirPooling error and classification accuracy]{\includegraphics[height=5.8cm]{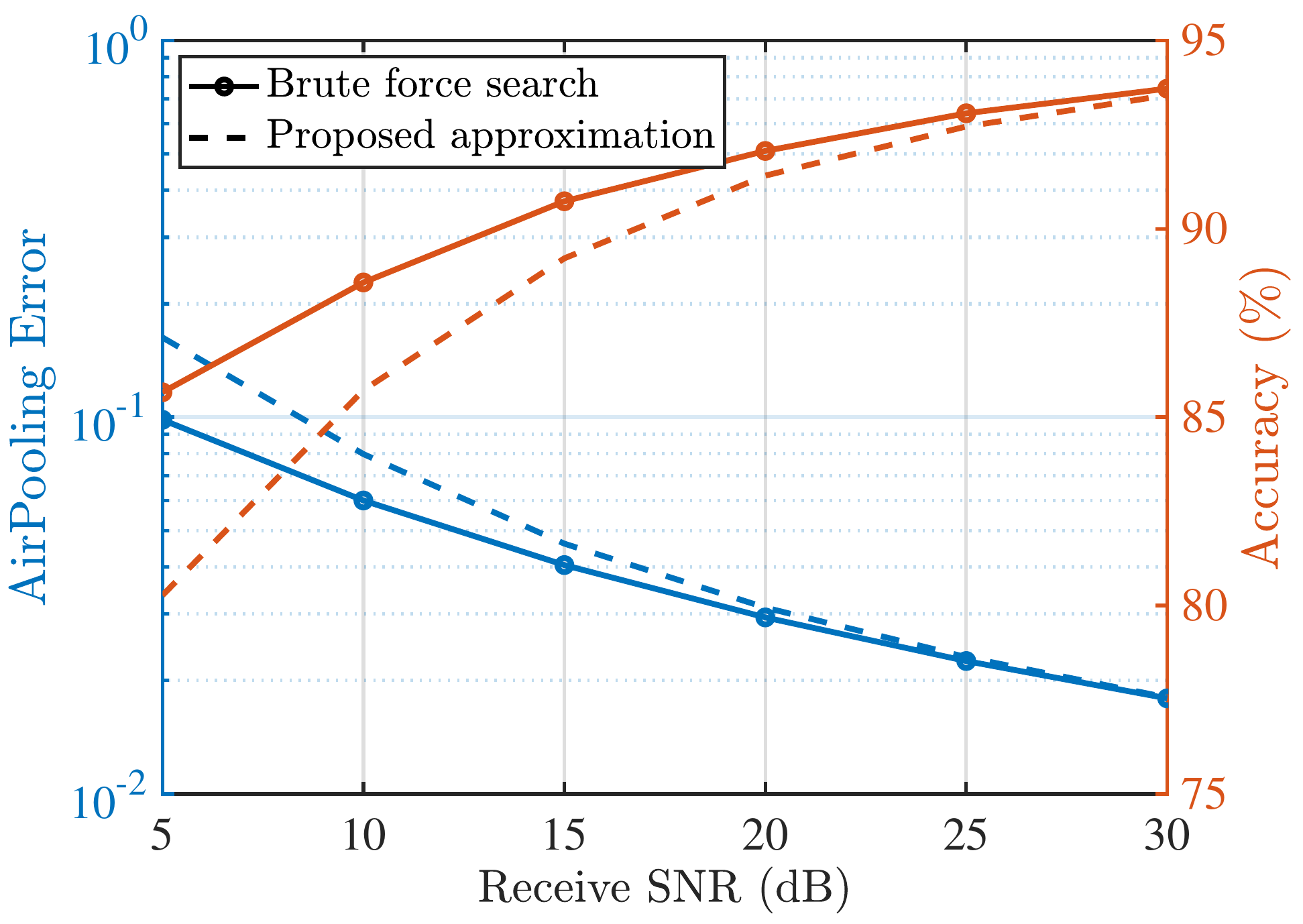}\label{fig: snn-max-rician-mse-acc}}
\hspace{0.5cm}
\subfigure[Optimal configuration parameter $\alpha$]{\includegraphics[height=5.8cm]{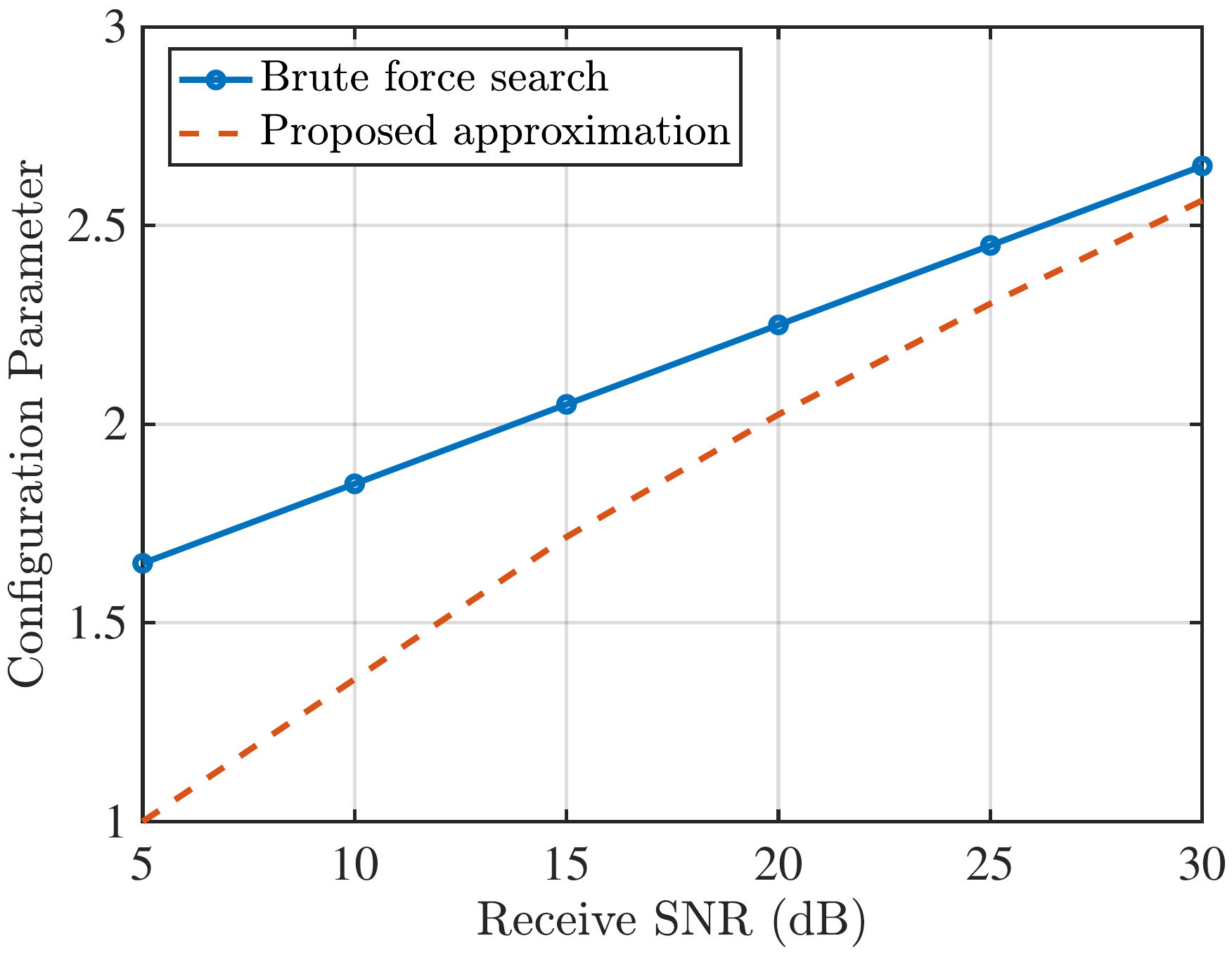}\label{fig: snn-max-rician-alpha}}
\caption{Performance comparison between the proposed sub-optimal method for setting the configuration parameter ($\alpha$) for Max-AirPooling and the optimal approach via a brute-force search in terms of (a) AirPooling error, classification accuracy, and (b) computed parametric value. }
\label{fig: snn-max-rician}

\end{figure*}

We evaluate the performance of the proposed near-optimal method for setting the configuration parameter for Max-AirPooling in Theorem 2. The synthetic feature dataset is used for its useful manipulability  and hence the SNN as the server classifier. The baseline approach is a brute-force linear search for the optimal parameter that minimizes the empirical AirPooling error. In Fig.~\ref{fig: snn-max-rician-mse-acc}, we plot the curves of AirPooling error and classification accuracy versus receive SNR. It can be observed that near-optimal performance is achieved by the proposed approach, which can be attributed to the tightness of the derived AirPooling-error upper bound on the classification error. In addition, the opposite trend of AirPooling error and accuracy curves is aligned with the derived relation betwen  AirPooling error and accuracy, which allows the former's minimization to substitute the latter's maximization. Fig.~\ref{fig: snn-max-rician-alpha} compares the configuration parameter values  obtained using the proposed sub-optimal method in Theorem 2 and brute-force search. It can be seen that there exists an optimality gap due to the use of bounds and approximation in the method, which, however,  narrows as the transmit SNR grows. 

\begin{figure}[t]
    \centering
    \includegraphics[scale=0.42]{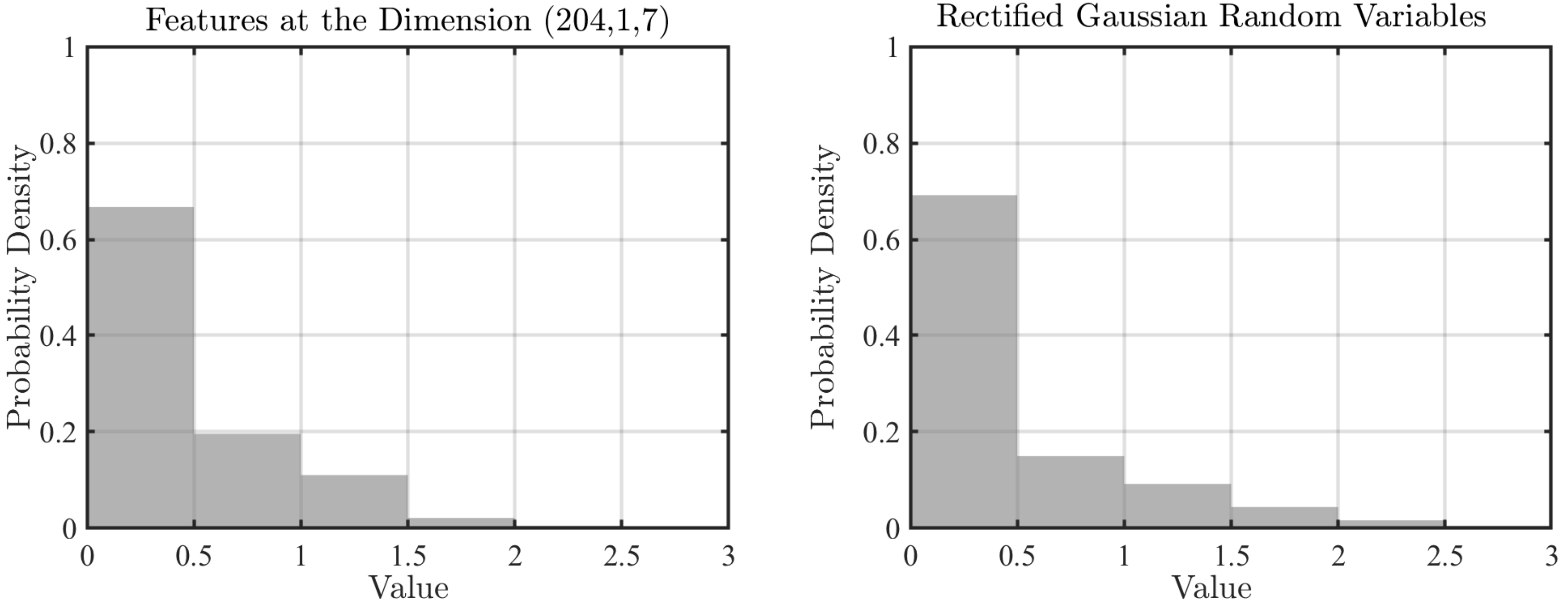}
    \caption{
    Distributions of the feature dimension $(247, 1, 7)$ on ModelNet10 and rectified Gaussian random variables.}
    \label{fig: feature-hists}
\end{figure}

\subsection{Performance of Max-AirPooling}
We first evaluate the performance of Max-AirPooling. It is essential to address the issue of optimizing the configuration parameter for the MVCNN model and real-world datasets. The feature distributions on ModelNet10 and ShapeNet are not exactly the rectified Gaussian distribution~\eqref{eqn: rectified_gaussian_pdf} used to derive the control algorithms. 
However, features distributions on real-world datasets are well-known to be in the shape of \emph{sparse activation}, i.e., being (near) zero with high probability, and being high-magnitude with low probability~\cite{Sun2021NeurIPS}. In that sense, the rectified Gaussian distribution is similar to real-world distributions (see Fig.~\ref{fig: feature-hists} for illustration). Therefore, for Max-AirPooling we propose a simple yet effective \emph{linear} calibration to transform $\alpha^{*}$ optimized for the rectified Gaussian distribution (derived using Theorem~\ref{optimized_alpha_no_selection}) to that used for ModelNet or ShapeNet, $\alpha^{\dagger}=c_1 \alpha + c_2$. The constants $c_1$ and $c_2$ are  fine-tuned via minimizing the expected AirPooling error.

\subsubsection{Sensing Accuracy and AirPooling Error}
Consider  Max-AirPooling over Rician channels on ModelNet10 and ShapeNet. The curves of classification accuracy and AirPooling error against receive SNR levels are plotted in Figs.~\ref{fig: mn-max-rician}. We choose a sufficiently high quantization resolution $Q=32$ bits per feature (i.e., full precision) for digital air interface. The benchmark achieves a maximum accuracy $91.74\%$ on ModelNet and a maximum accuracy $94.35\%$ on ShapeNet. In comparison,  the accuracy loss of Max-AirPooling is not higher than $0.61\%$ and $0.84\%$ at a moderate SNR (e.g., $6$ dB), respectively, while it is as low as {$0.44\%$ and $0.46\%$} at a high SNR (e.g., $20$ dB). Such marginal loss demonstrates the robustness of Max-AirPooling. Next, 
comparing Figs.~\ref{fig: mn-max-rician}(a) and~\ref{fig: mn-max-rician}(b), one can observe that the accuracy is insensitive to changes on AirPooling error when it is lower than $0.11$ on ShapeNet. It can be explained using  the concept of classification margin that the margin in this model tolerates AirPooling error not larger than $0.11$. When the AirPooling error goes beyond $0.17$ as the receive SNR is lower than $5$ dB, the accuracy suffers from notable degradation. A similar observation can also be made on ModelNet10. The above observations show the effectiveness of using AirPooling error as a surrogate metric of accuracy for AirPooling control.

\begin{figure*}[t]
\centering
\subfigure[Classification Accuracy]{\includegraphics[height=5.8cm]{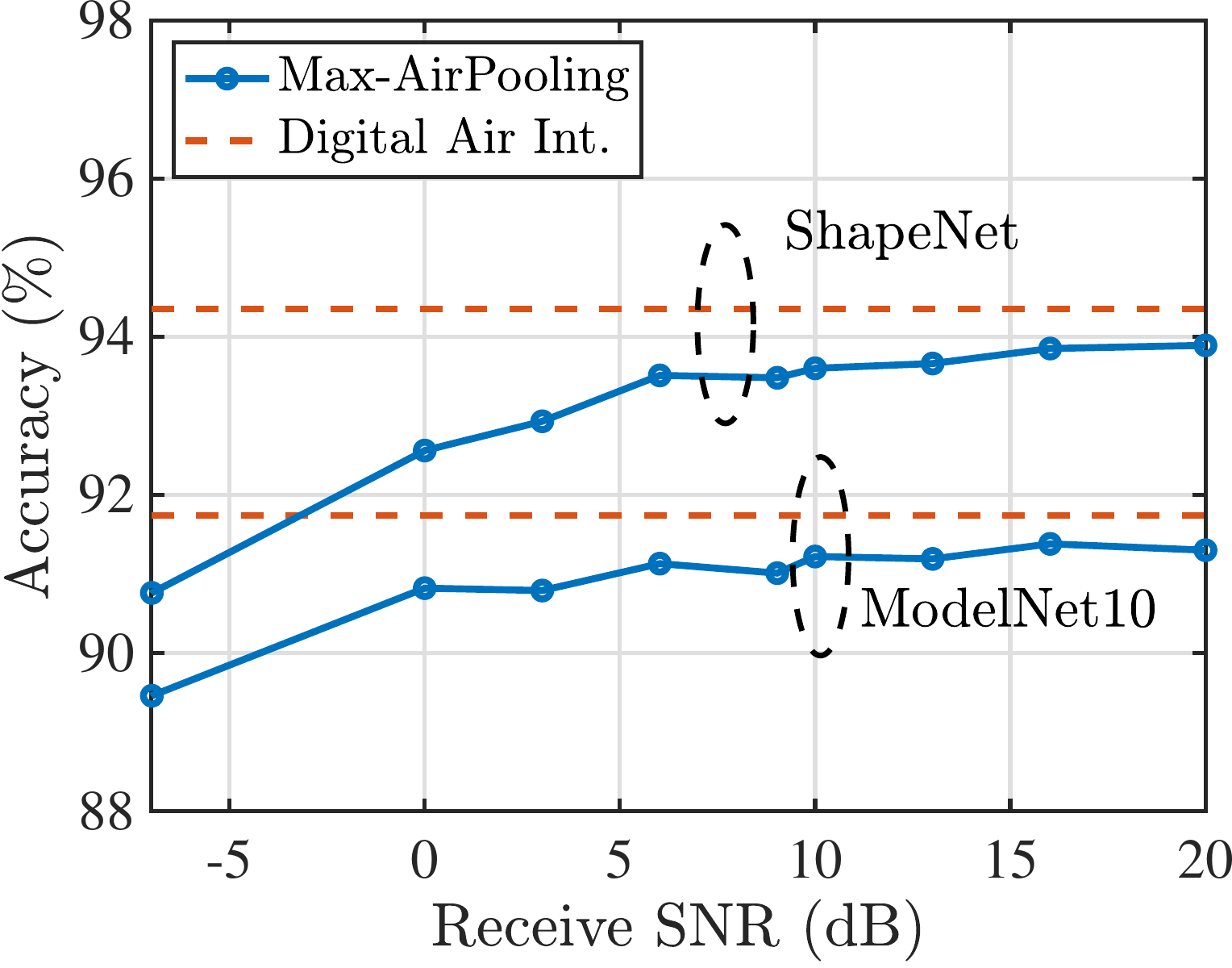}\label{fig: mn-max-rician-acc}}
\hspace{1cm}
\subfigure[AirPooling Error]{\includegraphics[height=5.8cm]{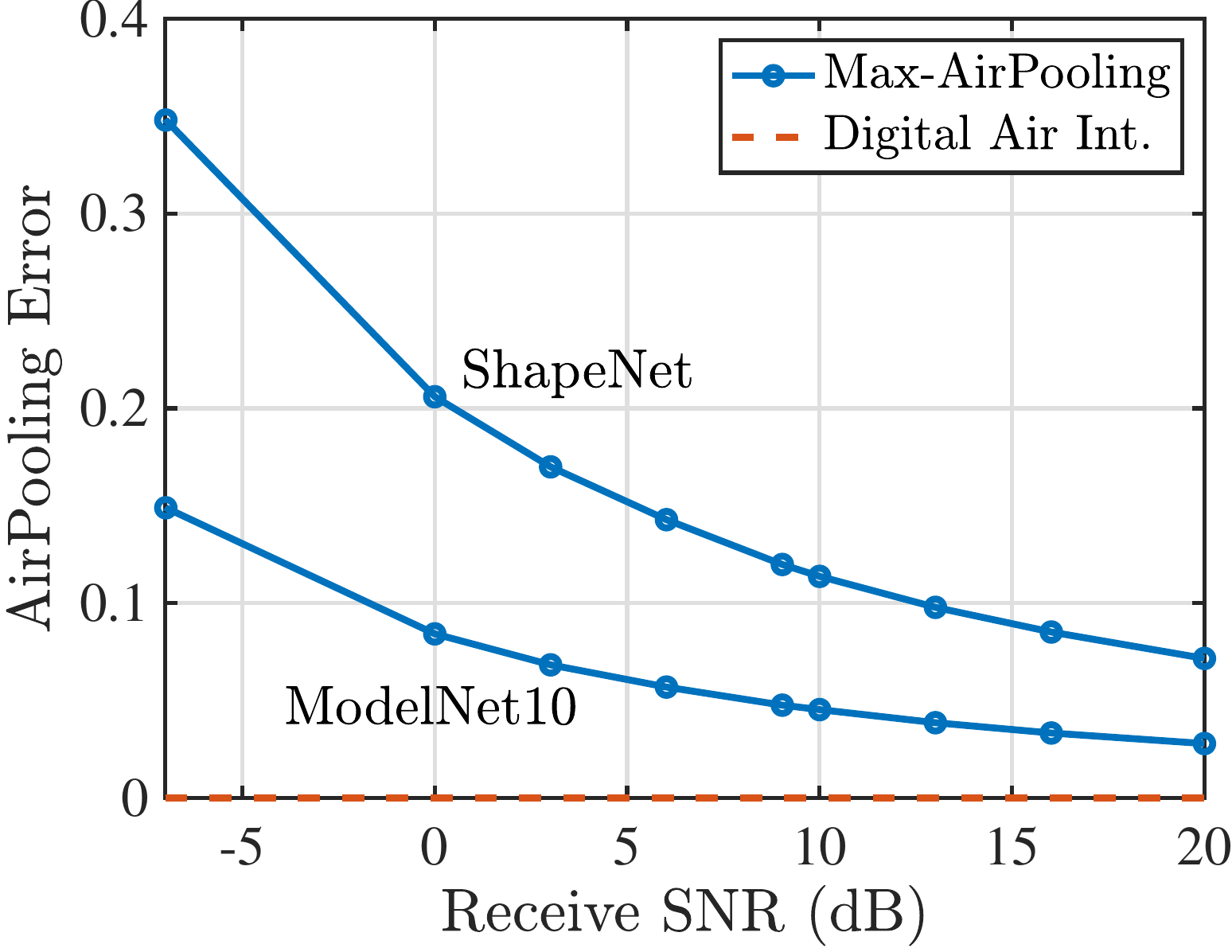}\label{fig: mn-max-rician-mse}}
\caption{Performance comparison between Max-AirPooling and the benchmark digital air interface at the resolution $Q=32$ bit over Rician channels on ModelNet10 and ShapeNet datasets for (a) inference accuracy and (b) AirPooling error. }
\label{fig: mn-max-rician}

\end{figure*}

\begin{figure*}[t]
\centering
\subfigure[ModelNet10]{\includegraphics[height=5.8cm]{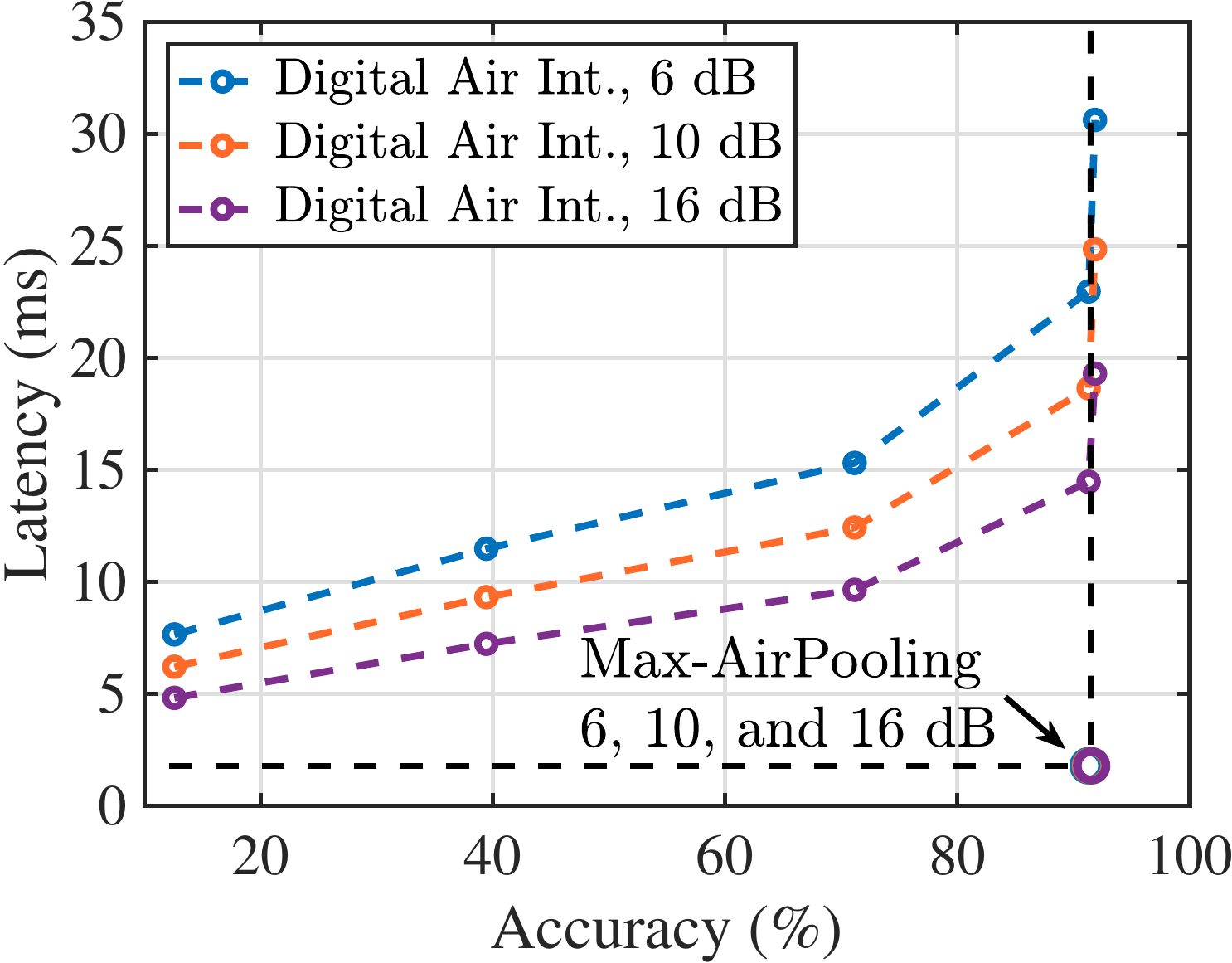}\label{fig: mn-max-rician-pareto}}
\hspace{1cm}
\subfigure[ShapeNet]{\includegraphics[height=5.8cm]{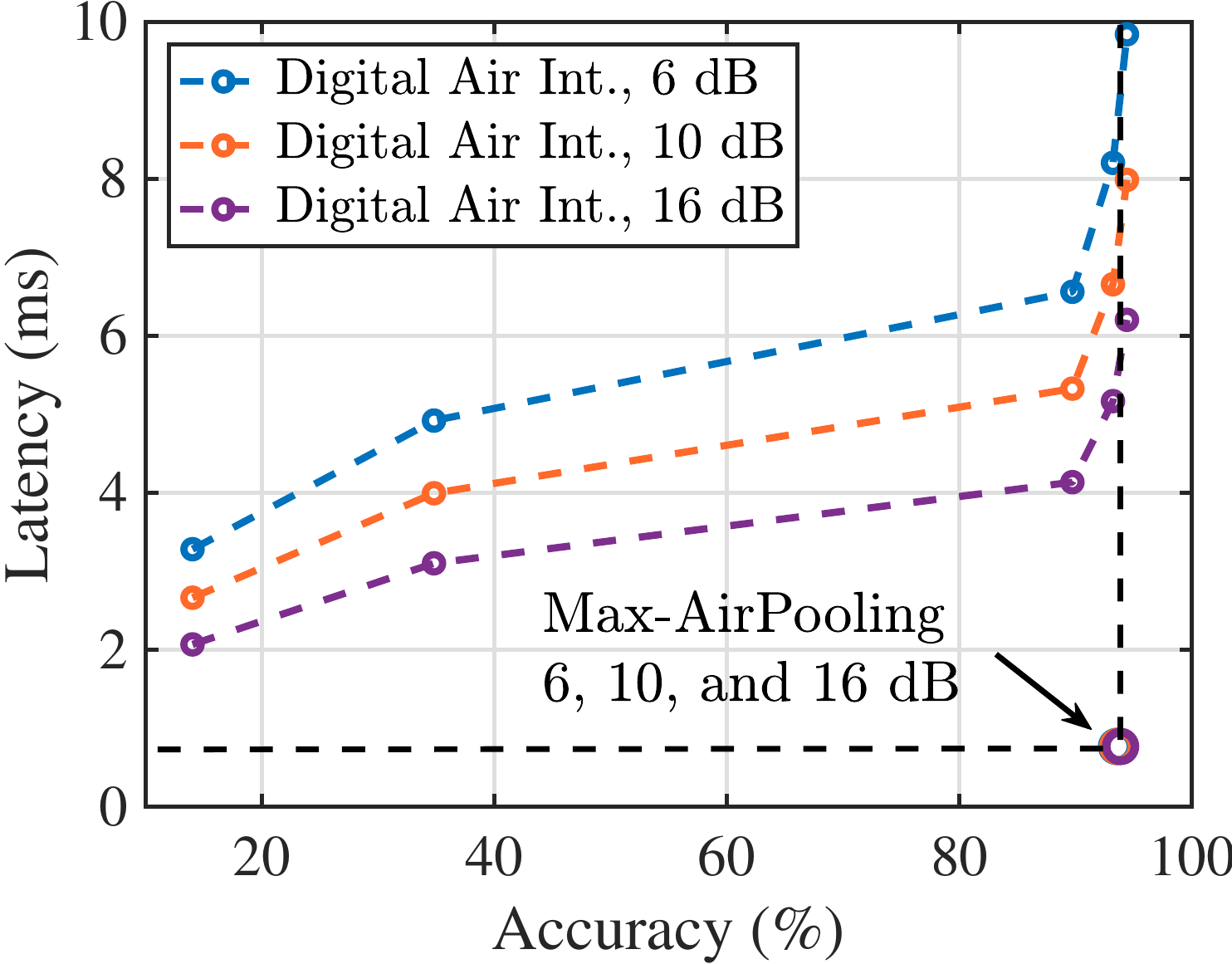}\label{fig: sn-max-rician-pareto}}
\caption{Comparison of communication latency (in \emph{millisecond} (ms)) between Max-AirPooling and the benchmarking scheme with varying target accuracies and receive SNR levels (in dB) on (a) ModelNet10 and (b) ShapeNet datasets. }
\label{fig: max-pareto}

\end{figure*}

\subsubsection{Communication Latency}
Define the communication latency as the transmission time (in milliseconds) required to achieve a target accuracy level. Extracted by VGG11 in MVCNN models, the size of feature tensors at each sensor is $512\times7\times7$ and $512\times4\times4$ for ModelNet10 and ShapeNet, respectively. By neglecting the approximately constant feature dimensions (i.e., the feature dimensions with variance lower than ${10}^{-12}$ and, e.g., those being always close-to-zero), 
the number of features to be aggregated over-the-air is $N=17911$ for ModelNet10 and $N=7675$ for ShapeNet. While attaining comparable accuracies as the traditional  digital air interface,
the superiority of Max-AirPooling is reflected in dramatic latency reduction as shown by the results in Fig.~\ref{fig: max-pareto}. For Max-AirPooling, the latency is fixed as given by $L_{\sf AirPooling} = \frac{N}{B}$. In the case of digital air interface, the latency not only depends on the power budget but also varies with channel fading; the expected latency is taken over the channel distribution.
Consider a target accuracy of $91\%$ which requires a resolution of $Q=6$ bit per feature for digital air interface. The resultant latencies are $22.90,\ 18.64,\ \text{and}\ 14.47$ milliseconds given receive SNRs at $6$ dB, $10$ dB, and $16$ dB, respectively.  Max-AirPooling reduces the latency to $1.79$ millisecond, achieving a reduction ratio of $12$, $10$, and $8$ with respect to the digital counterpart, respectively. Similar results can be observed on ShapeNet (see  Fig.~\ref{fig: sn-max-rician-pareto}). For instance, $10$-time latency reduction can be obtained at the  receive SNR of $6$ dB and a target accuracy of $93\%$.  Moreover, this latency advantage of Max-AirPooling is especially large in the low transmit-power regime. In this regime, AirPooling benefits more significantly from the model robustness (due to classification margin) against noise and the noise suppression capability of the aggregation operation.

\begin{figure*}[t]
\centering
\subfigure[ModelNet10]{\includegraphics[height=5.8cm]{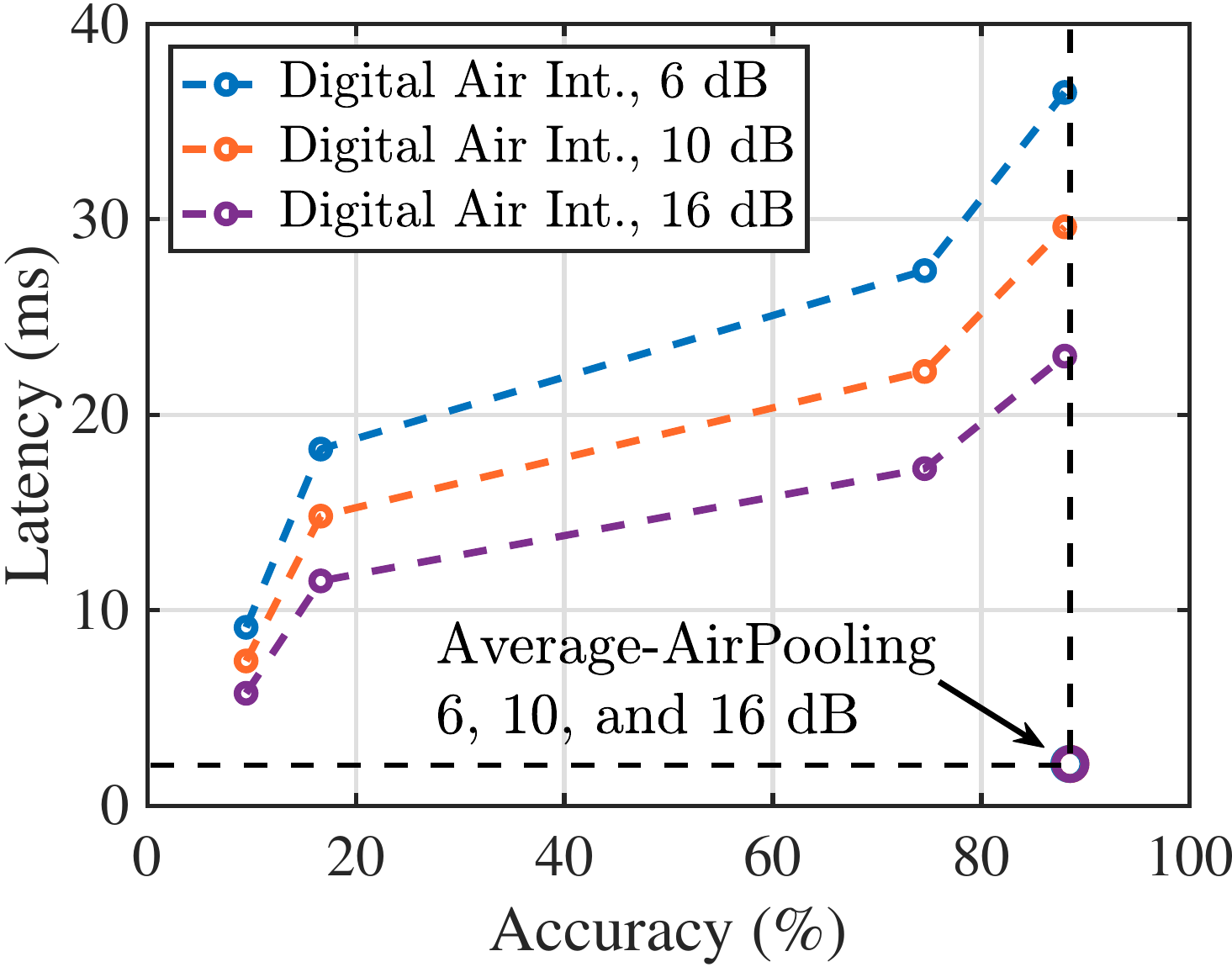}\label{fig: mn-avg-rician-pareto}}
\hspace{1cm}
\subfigure[ShapeNet]{\includegraphics[height=5.8cm]{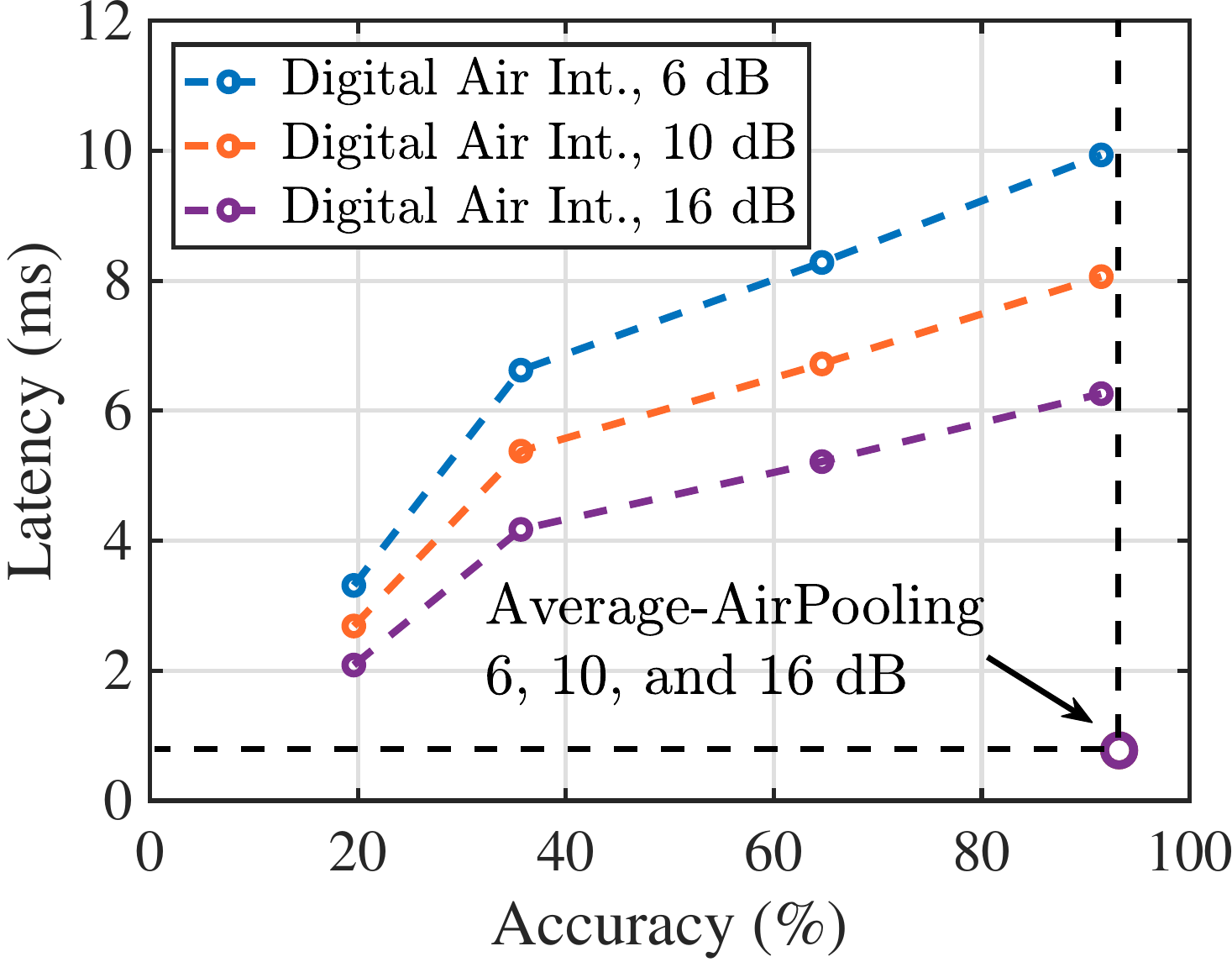}\label{fig: sn-avg-rician-pareto}}
\caption{Comparison of communication latency (in millisecond (ms)) between Average-AirPooling and the benchmarking scheme with varying target accuracies and receive SNR levels (in dB) on (a) ModelNet10 and (b) ShapeNet datasets. }
\label{fig: mn-pareto}

\end{figure*}

\subsection{Performance of Average-AirPooling}
As for Average-AirPooling, the configuration parameter is fixed as $\alpha^{*}=1$ (see Proposition~\ref{proposition: avg-airpooling}). Similar to the Max-AirPooling counterpart, Average-AirPooling achieves comparable accuracies as the digital air interface on both the ModelNet10 and ShapeNet datasets. As shown in Figs.~\ref{fig: mn-avg-rician-pareto} and~\ref{fig: sn-avg-rician-pareto}, Average-AirPooling achieves $88.44\%$ and $93.21\%$ with receive SNR as low as $6$ dB on ModelNet10 and ShapeNet, respectively. These indicate a performance loss of less than $1\%$ with respect to the ideal, errorless pooling results in Table~\ref{table: test_acc_modelnets} (or equivalently the digital counterpart with a high quantization resolution). The experiments on communication latency  are also conducted  for Average-AirPooling. For the digital air interface, varying the quantization resolution yields a set of pairs of latency and accuracy given certain receive SNR levels. Targetting on $88\%$ accuracy on ModelNet10, Average-AirPooling reduces the latency by more than $10$ times. On ShapeNet, Average-AirPooling reduces latency by a factor of $8$ to $10$ for the target  accuracy of  $93\%$. The above results also validate the effectiveness of Average-AirPooling and its superiority in air latency for distributed sensing like Max-AirPooling.

\section{Concluding Remarks}
\label{sec: conclusion}
In this paper, we have presented the framework of AirPooling for communication-efficient distributed sensing. AirPooling exploits controlled vector norms and the waveform superposition property of wireless channels to support a wide variety of fusion functions, especially the popular max- and average-pooling over multiple access channels. The fundamental tradeoff between the suppression of function-approximation error and mitigating channel noise perturbation is revealed and optimized. Comprehensive experiments justify the efficiency gain of AirPooling in terms of the latency to achieve a target accuracy or vice versa.  

From a high-level perspective, this work initializes an important step toward task-oriented air interfaces for distributed sensing, opening a line of follow-up studies. One direction centers on spatially multiplexing AirPooling via advanced \emph{multiple-input multiple-output} (MIMO) techniques. It involves the joint optimization of precoders at sensors and beamformers at servers. In view of E2E performance metrics like sensing accuracy and network lifetime, another direction is to design task-oriented \emph{radio-resource management} (RRM) techniques for AirPooling. The goal aims at maximizing the E2E performance, especially for sensing networks with heterogeneous sensor capabilities and classification margins.

\appendix

\subsection{Proof of Lemma~\ref{lemma: post_processing_factor}}
\label{proof: best_beta}
For Average-AirPooling, setting $\beta=K$ and $\alpha=1$ yields $\tilde{g}=g_{\mathsf{avg}}$ and is hence optimal. For Max-AirPooling, the noise-free AirPooling error is given by
\begin{align}
    \mathbb{E}[(\tilde{g}-g_{\mathsf{max}})^2]
    & = \mathbb{E}\left[\left(\frac{1}{\beta^{1/\alpha}}\Vert\bff\Vert_{\alpha}-f_{\max}\right)^2\right]\notag\\
    & \triangleq \mathbb{E}\left[\left(u\Vert\bff\Vert_{\alpha}-f_{\max}\right)^2\right]\notag\\
    & = u^2\mathbb{E}\left[\Vert\bff\Vert_{\alpha}^2\right]-2u\mathbb{E}\left[\Vert\bff\Vert_{\alpha}f_{\max}\right]+\mathbb{E}\left[f_{\max}^2\right]\label{eq: minimize_amse}
\end{align}
 where the third equality is by substitution with  $u=\left(\frac{1}{\beta}\right)^{1/\alpha}$. It is then obvious that $\mathbb{E}[(\tilde{g}-g_{\mathsf{max}})^2]$ achieves global minimum at $u=u^*$, where $u^*$ is given by
 \begin{equation}\label{eq: optimal_u}
     u^*=\frac{\mathbb{E}\left[\Vert\bff\Vert_{\alpha}f_{\max}\right]}{\mathbb{E}\left[\Vert\bff\Vert_{\alpha}^2\right]},
 \end{equation}
 which proves the optimality of $\beta^*$.
 Further, noting that $0\leq f_k\leq f_{\max}$, the following inequalities hold
 \begin{equation}
    f_{\max} \leq \Vert\bff\Vert_{\alpha} \leq K^{1/\alpha} f_{\max}.
 \end{equation}
 Thus, by multiplying with $\Vert\bff\Vert_{\alpha}$ and taking expectation on both sides of the inequalities, we obtain the following inequality
 \begin{equation}
    \mathbb{E}\left[f_{\max}\Vert\bff\Vert_\alpha  \right]\leq\mathbb{E}\left[\Vert\bff\Vert^2_{\alpha}\right] \leq K^{1/\alpha} \mathbb{E}\left[f_{\max}\Vert\bff\Vert_\alpha\right].
 \end{equation}
 Substituting the inequality into \eqref{beta_n} gives the said bound of $\beta^*$. This completes the proof. 

\subsection{Proof of Theorem~\ref{theorem: approximation_capability}}
\label{app: proof_approximation_capability}
To begin with, we note that the chosen post-processing parameter for max-pooling as in~\eqref{eqn: approximation_capability} satisfies $1\leq\beta^*\leq K$. Since $0\leq f_k\leq \max{\{f_k\}}$ for all $k=1,\ldots,K$, an upper bound of $\tilde{g}$ is hence given by
    \begin{equation}
        \tilde{g} = \left(\frac{1}{\beta^*}\sum_{k=1}^K f^{\alpha}_k\right)^{1/{\alpha}}\leq \left(\sum_{k=1}^K f^{\alpha}_k\right)^{1/{\alpha}} \leq (K\max{\{f_k\}}^\alpha)^{1/\alpha}=K^{1/\alpha}\max{\{f_k\}}.
    \end{equation}
    For the same reason, a lower bound is given by
    \begin{equation}
        \tilde{g} = \left(\frac{1}{\beta^*}\sum_{k=1}^K f^{\alpha}_k\right)^{1/{\alpha}}\geq \left(\frac{1}{K}\sum_{k=1}^K f^{\alpha}_k\right)^{1/{\alpha}} \geq \left(\frac{1}{K}\max{\{f_k\}}^\alpha\right)^{1/\alpha}=\frac{1}{K^{1/\alpha}}\max{\{f_k\}}.
    \end{equation}
It can be observed that both the upper and lower bounds of $\tilde{g}$ converge to $g_{\mathsf{max}}=\max{\{f_k\}}$ given $\alpha\rightarrow\infty$. The well-known sandwich theorem then yields the first relation. Furthermore, by comparing $ \tilde{g} = \left(\frac{1}{\beta}\sum_{k=1}^K f^{\alpha}_k\right)^{1/{\alpha}}$ to \eqref{perfect-pooling}, the second relation obviously holds. This completes the proof.

\subsection{Proof of Lemma~\ref{lemma: acc_lb}}\label{proof: acc_lb}
The first inequality follows directly from the fact that the perturbed feature vector $\hat{\bg}$ is correctly classified if $d(\bg,\hat{\bg})=\Vert\hat{\bg}-\bg\Vert_2=\Vert\bee\Vert_2<\Delta$. Using the well-known \emph{Markov's inequality}, the second inequality is derived as 
\begin{align}
    R_0\mathsf{Pr}[\Vert\bee\Vert_2<\Delta]&=R_0\mathsf{Pr}[\Vert\bee\Vert_2^2<\Delta^2]\notag\\
    &\geq R_0\left(1-\frac{\mathbb{E}[\Vert\bee\Vert_2^2]}{\Delta^2}\right).
\end{align}
Note that $\mathbb{E}[\Vert\bee\Vert_2^2]=\sum_{n=1}^N \mathbb{E}[|\hat{g}_n-g_n|^2]=\sum_{n=1}^N D_n$. This completes the proof.

\subsection{Proof of Lemma~\ref{proposition_decompose}}
\label{app: proof_proposition}
In Lemma~\ref{proposition_decompose}, the upper bound for max-pooling is established as
\begin{align}
    D(\alpha) & = \mathbb{E}[|\hat{g}-g|^2] \notag\\
    & = \mathbb{E}[|\hat{g}-\tilde{g}-({g}-\tilde{g})|^2] \notag\\
    & = \mathbb{E}[|\hat{g}-\tilde{g}|^2+2(\hat{g}-\tilde{g})(\tilde{g}-g)+|\tilde{g}-g|^2]\notag\\
    & \leq 2\mathbb{E}[|\hat{g}-\tilde{g}|^2]+2\mathbb{E}[|\tilde{g}-g|^2]\notag\\
    & = 2 \left[ D_{\sf chan}(\alpha) + D_{\sf appr}(\alpha)\right]\notag,
\end{align}
For average-pooling, the derivation is similar except that $\xi=\hat{g}-\tilde{g}$ and $(\tilde{g}-g)$ are independent and zero mean and thus $\mathbb{E}[(\hat{g}-\tilde{g})(\tilde{g}-g)]=0$. We then have $D(\alpha)=\mathbb{E}[|\hat{g}-\tilde{g}|^2]+\mathbb{E}[|\tilde{g}-g|^2]=   D_{\sf chan}(\alpha) + D_{\sf appr}(\alpha)$ for average-pooling.

\subsection{Proof of Lemma~\ref{lemma: mse_max_no_selection}}
\label{app: proof_lemma_mse_max_noselection}
The term $D_{\sf chan} (\alpha) = \mathbb{E}[|\hat{g}-\tilde{g}|^2]$ can be bounded as
\begin{align}
    D_{\sf chan} (\alpha)  
    & = \mathbb{E}\left[\left|\left[\frac{1}{\beta}\left(\sum_{k=1}^K f^{\alpha}_k+\xi\right)^+\right]^{1/{\alpha}}-\left(\frac{1}{\beta}\sum_{k=1}^K f^{\alpha}_k\right)^{1/{\alpha}}\right|^2\right] \notag\\
    & \leq \mathbb{E}\left[\left(\frac{1}{\beta}|\xi|^{1/\alpha}\right)^2\right] \notag \\
    & \leq \mathbb{E}[(|\xi|^{1/\alpha})^2] \notag \\
    & = \left(\frac{\sigma^2\nu^2}{P_\mathsf{rx}}\right)^{1/\alpha}\mathbb{E}[(|w|^2)^{1/\alpha}] \notag \\
    & \leq \left(\frac{\sigma^2\nu^2}{P_\mathsf{rx}}\right)^{1/\alpha}\triangleq {\delta},
\end{align}
where $w\sim \mathcal{N}(0,1)$, the first inequality is established by proving
    $\left\{\left[(a+b)^+\right]^{\frac{1}{\alpha}}-a^{\frac{1}{\alpha}}\right\}^2\leq|b|^{\frac{1}{\alpha}}$
for $a\geq 0$ and $b\in \mathbb{R}$ with the details omitted for brevity, and the last inequality is via Jensen's inequality. 
Next, we derive the upper bound of $D_{\sf appr}(\alpha)$ in Max-AirPooling, given that $\beta=\beta^*$ as defined in Proposition 1. By substituting \eqref{eq: optimal_u} into \eqref{eq: minimize_amse}, $D_{\sf appr}(\alpha)$ is bounded as
\begin{align}
    D_{\sf appr}(\alpha)&=\mathbb{E}\left[f_{\max}^2\right]-\frac{\mathbb{E}\left[\Vert\bff\Vert_{\alpha}f_{\max}\right]^2}{\mathbb{E}\left[\Vert\bff\Vert_{\alpha}^2\right]}\notag\\
    &\leq\mathbb{E}\left[f_{\max}^2\right]-\frac{\mathbb{E}\left[f^2_{\max}\right]\mathbb{E}\left[\Vert\bff\Vert_{\alpha}f_{\max}\right]}{\mathbb{E}\left[\Vert\bff\Vert_{\alpha}^2\right]}\notag\\
    &\leq \mathbb{E}\left[f_{\max}^2\right]\left(1-\frac{\mathbb{E}\left[\Vert\bff\Vert_{\alpha}f_{\max}\right]}{K^{\frac{1}{\alpha}}\mathbb{E}\left[\Vert\bff\Vert_{\alpha}f_{\max}\right]}\right) = \mathbb{E}\left[f_{\max}^2\right]\left(1-K^{-\frac{1}{\alpha}}\right)\triangleq \epsilon_\mathsf{m}.
\end{align}
The proof for Average-AirPooling is trivial and omitted for brevity. Then, using Lemma~\ref{proposition_decompose} the said AirPooling error upper bound can be obtained.

\subsection{Proof of Lemma~\ref{lemma: delta_m_rate}}\label{proof_lemma_delta_m_asym}
Starting from \eqref{eqn: delta_m_gamma}, we have
\begin{align}
    {\delta}
    &= 2\left[\frac{\sigma^2}{2\bar{P}\sqrt{\pi}}\Gamma\left(\frac{2\alpha+1}{2}\right)\right]^{1/\alpha}\left[1-\frac{\Gamma^2\left(\frac{\alpha+1}{2}\right)}{2\sqrt{\pi}\Gamma\left(\frac{2\alpha+1}{2}\right)}\right]^{1/\alpha}\notag \\
    &\triangleq 2\left[\frac{\sigma^2}{2\bar{P}\sqrt{\pi}}\Gamma\left(\frac{2\alpha+1}{2}\right)\right]^{1/\alpha}\left[1-g(\alpha)\right]^{1/\alpha},\label{eqn: proof_lemma5_1}
\end{align}
where $g(\alpha)=O\left(\frac{1}{2^\alpha}\right)$ can be obtained using the Stirling's approximation of Gamma function, i.e., $\Gamma(z)= \sqrt{2\pi}z^{z-1/2}e^{-z}\left(1+O\left(\frac{1}{z}\right)\right)$. Expanding $\left[1-g(\alpha)\right]^{1/\alpha}$ with Taylor series gives
\begin{equation}
    \left[1-g(\alpha)\right]^{1/\alpha} = 1-\frac{g(\alpha)}{\alpha}+o\left(\frac{g(\alpha)}{\alpha}\right)= 1+O\left(\frac{1}{ \alpha 2^\alpha}\right).\label{eqn: proof_lemma5_2}
\end{equation}
Substituting \eqref{eqn: proof_lemma5_2} into \eqref{eqn: proof_lemma5_1} gives
\begin{align}
    {\delta}
    &= 2\left[\frac{\sigma^2}{2\bar{P}\sqrt{\pi}}\Gamma\left(\frac{2\alpha+1}{2}\right)\right]^{1/\alpha}\left[1+O\left(\frac{1}{ \alpha 2^\alpha}\right)\right]\notag\\
    &= 2\left[\frac{\sigma^2}{2\bar{P}\sqrt{\pi}}\Gamma\left(\frac{2\alpha+1}{2}\right)\right]^{1/\alpha} + O \left(\frac{1}{2^\alpha}\right),
\end{align}
where the second equality is due to $\left[\frac{\sigma^2}{2\bar{P}\sqrt{\pi}}\Gamma\left(\frac{2\alpha+1}{2}\right)\right]^{1/\alpha}=O(\alpha)$ as obtained via the Stirling's approximation. This proves the first equality of Lemma 5. From the second equality, we proceed with
\begin{align}
    {\delta}
    &= 2\left\{\frac{\sigma^2}{2\bar{P}\sqrt{\pi}}\sqrt{2\pi}\left(\alpha+\frac{1}{2}\right)^{\alpha}e^{-(\alpha+1/2)}\left[1+O\left(\frac{1}{\alpha}\right)\right]\right\}^{1/\alpha} + O \left(\frac{1}{2^\alpha}\right)\nonumber\\
    &=2e^{-1}\left(\frac{\sigma^2}{\sqrt{2}\bar{P}}\right)^{1/\alpha}\left(\alpha+\frac{1}{2}\right)e^{-1/2\alpha}\left[1+O\left(\frac{1}{\alpha}\right)\right]^{1/\alpha} + O \left(\frac{1}{2^\alpha}\right)\nonumber\\
    &=2e^{-1}\left(\frac{\sigma^2}{\sqrt{2}\bar{P}}\right)^{1/\alpha}\left(\alpha+\frac{1}{2}\right)\left[1-\frac{1}{2\alpha}+O\left(\frac{1}{\alpha^2}\right)\right]\left[1+O\left(\frac{1}{\alpha^2}\right)\right] + O \left(\frac{1}{2^\alpha}\right)\nonumber\\
    &=2e^{-1}\left(\frac{\sigma^2}{\sqrt{2}\bar{P}}\right)^{1/\alpha}\left[\alpha+O\left(\frac{1}{\alpha}\right)\right] + O \left(\frac{1}{2^\alpha}\right)\nonumber\\
    &=2e^{-1}\left(\frac{\sigma^2}{\sqrt{2}\bar{P}}\right)^{1/\alpha}\alpha + O \left(\frac{1}{\alpha}\right),
\end{align}
where the third equality is by Taylor expansion. This completes the proof. 

\subsection{Proof of Theorem~\ref{optimized_alpha_no_selection}}\label{proof_optimal_alpha}
Consider equation \eqref{eqn: noselection_suboptimal_equation}. If it has a unique root $\alpha^*$, then via straightforward analysis of the first-order derivative this root minimizes $\tilde{\delta}_{\sf m}+{\epsilon}_{\sf m}$. Let $u=\frac{C}{\alpha}+A$ with $A<u\leq C+A$, where $C \triangleq \log \left(\frac{\sqrt{2}\bar{P}}{K\sigma^2}\right)$ and $A\triangleq\frac{C}{C+\log K}$. Substituting $\alpha=\frac{C}{u-A}$ into equation \eqref{eqn: noselection_suboptimal_equation} allows us to simplify the expression as
\begin{equation}\label{eqn: u_eqn}
    A\exp\left(u-A\right)\left(u-2A+\frac{A^2}{ u}\right)=\frac{2e^{-1}C^2}{\mathbb{E}[f_{\max}^2\vert K]\log{K}}. 
\end{equation}
To further simplify the equation, note that $\frac{A}{2u}=\frac{\alpha}{2(C+\log K +\alpha)}\ll 1$
holds in the high-SNR regime because $C+\log K=\log \left(\frac{\sqrt{2}\bar{P}}{\sigma^2}\right)$ is greater than $\alpha$ by at least an order of magnitude. Therefore, $2A\gg\frac{A^2}{ u}$ and we can omit the latter. A approximated form of \eqref{eqn: u_eqn} is then obtained as
\begin{equation}\label{eqn: u_approx_eqn}
    A\exp\left(u-A\right)\left(u-2A\right)=\frac{2e^{-1}C^2}{\mathbb{E}[f_{\max}^2\vert K]\log{K}}. 
\end{equation}
Next, let $v=u-2A$ with $-A<v\leq C-A\triangleq v_\mathsf{m}$. Inserting $u=v+2A$ into \eqref{eqn: u_approx_eqn} yields
\begin{equation}\label{eqn: v_eqn}
    e^v v = \frac{2C(C+\log K)}{\exp{\left(1+A\right)} \mathbb{E}[f_{\max}^2\vert K]\log{K}}\triangleq L. 
\end{equation}
Note that $J_0(v)\triangleq e^v v$ is monotonically increasing in $(0,v_{\sf m})$ and $J(0)=0$.  Also, we have $J_0(v_\mathsf{m})>L$ for any $C>\frac{1}{2}\log2$ and $K\geq 4$.
Therefore, equation \eqref{eqn: v_eqn} has a unique root on $(0,v_{\sf m})$. The root can be expressed with the principal branch of Lambert $W$ function, which is the inverse function of $J_0(v)$, as given by $v^*=W_0(L)$. Finally,  evaluating $\alpha^*=\frac{C}{v^*+\frac{C}{C+\log K}}$ yields the desired result.

\bibliographystyle{IEEEtran}

\end{document}